%% file: main.tex
\documentclass[a4paper,UKenglish]{lipics-v2021}
\usepackage{amssymb}
\usepackage{color}
\usepackage{nicefrac}
\usepackage{graphicx}
\usepackage{amsthm}
\usepackage{amsmath}

\usepackage[colorinlistoftodos]{todonotes}
\usepackage{xargs}  
\newcommandx{\pierre}[2][1=]{\todo[linecolor=blue,backgroundcolor=blue!25,bordercolor=blue,#1]{\tiny Pierre: #2}}

\newcommandx{\pedro}[2][1=]{\todo[linecolor=green,backgroundcolor=green!25,bordercolor=blue,#1]{\tiny Pedro: #2}}
\newcommandx{\revision}[1]{{\color{black}{#1}}}

\nolinenumbers 

\hideLIPIcs 

\author{Pierre Fraigniaud}{IRIF, Universit\'e Paris Cit\'e and CNRS, France.}{pierre.fraigniaud@irif.fr}{}{Additional support for ANR projects QuData and DUCAT.}
\author{Fr\'ed\'eric Mazoit}{LaBRI, Universit\'e de Bordeaux, France}{frederic.mazoit@labri.fr}{}{}
\author{Pedro Montealegre}{Facultad de Ingenier\'ia y Ciencias, Universidad Adolfo Ib\'a\~nez, Santiago, Chile.}{p.montealegre@uai.cl}{}{This work was supported by Centro de Modelamiento Matem\'atico (CMM), FB210005, BASAL funds for centers of excellence from ANID-Chile, and ANID-FONDECYT 1230599. }
\author{Ivan Rapaport}{DIM-CMM (UMI 2807 CNRS), Universidad de Chile, Chile.}{rapaport@dim.uchile.cl}{}{This work was supported by Centro de Modelamiento Matem\'atico (CMM), FB210005, BASAL funds for centers of excellence from ANID-Chile, and  ANID-FONDECYT 1220142. }
\author{Ioan Todinca}{LIFO, Universit\'e d'Orl\'eans and INSA Centre-Val de Loire, France.}{ioan.todinca@univ-orleans.fr}{}{}

\authorrunning{P. Fraigniaud et al.} 

\Copyright{Pierre Fraigniaud, Fr\'ed\'eric Mazoit, Pedro Montealegre, Ivan Rapaport and Ioan Todinca} 



\title{Distributed Certification for Classes of Dense Graphs}

\newcommand{\MSO}{\mathrm{MSO}}
\newcommand{\id}{\mathsf{id}}
\newcommand{\cw}{\mathsf{cw}}
\newcommand{\nlcw}{\mathsf{nlcw}}
\newcommand{\NLC}{\mathrm{NLC}}

\newcommand{\cC}{\mathcal{C}}
\newcommand{\recolor}{\mathsf{recolor}}

\begin{document}

\maketitle

\input{abstract.tex}

\input{introduction.tex}

\input{model.tex}

\input{highlevel.tex}

\input{cograph.tex}

\input{cliquewidth.tex}

\input{optimization.tex}

\input{conclusion.tex}


\bibliography{biblio-PLS-CW-MSO}


\end{document}

%% file: abstract.tex

\begin{abstract}
A \emph{proof-labeling scheme} (PLS) for a  boolean predicate~$\Pi$ on labeled graphs is a mechanism used for certifying the legality with respect to~$\Pi$ of global network states   in a distributed manner. In a PLS, a \emph{certificate} is assigned to each processing node of the network, and the nodes are in charge of checking that the collection of certificates forms a global proof that the system is in a correct state, by exchanging the certificates once, between neighbors only. The main measure of complexity is the \emph{size} of the certificates. Many PLSs have been designed for certifying specific predicates, including cycle-freeness, minimum-weight spanning tree, planarity, etc. 

In 2021, a breakthrough has been obtained, as a ``meta-theorem'' stating that a large set of properties have compact PLSs in a large class of networks. Namely, for every  $\MSO_2$ property~$\Pi$ on labeled graphs, there exists a PLS for~$\Pi$ with $O(\log n)$-bit certificates for all graphs of bounded \emph{tree-depth}. This result has been extended to the larger class of  graphs with bounded \emph{tree-width}, using certificates on $O(\log^2 n)$ bits. 

We extend this result even further, to the larger class of graphs with bounded \emph{clique-width}, which, as opposed to the other two aforementioned classes, includes dense graphs. We show that, for every $\MSO_1$ property~$\Pi$  on labeled graphs, there exists a PLS for $\Pi$ with $O(\log^2 n)$-bit certificates for all graphs of bounded clique-width. As a consequence, certifying families of graphs such as distance-hereditary graphs and (induced) $P_4$-free graphs (a.k.a., cographs) can be done using a PLS with $O(\log^2 n)$-bit  certificates, merely because each of these two classes can be specified in $\MSO_1$. In fact, we show that certifying $P_4$-free graphs can be done with certificates on $O(\log n)$ bits only. This is in contrast to the class of $C_4$-free graphs (which \revision{does} not have bounded clique-width) which requires $\tilde{\Omega}(\sqrt{n})$-bit certificates.
\end{abstract}

%% file: introduction.tex

\section{Introduction}

Checking whether a distributed system is in a legal global state with respect to some boolean predicate occurs in several domains  of distributed computing, including the following.
\begin{itemize}
\item Fault-tolerance: the \revision{occurrence} of faults may turn the system into an illegal state that needs to be detected for allowing the system to return to a legal state.
\item The use of subroutines as black boxes: some of these subroutines may contain bugs, and produce incorrect outputs that need to be checked before use in the protocol calling the subroutines. 
\item Algorithm design for specific classes of systems: an algorithm dedicated to some specific class of networks (e.g., algorithms for trees, or for planar networks) may cause deadlocks or live-locks whenever running on a network outside the class. The membership to the class needs to be checked before running the algorithm. 
 \end{itemize}
In all three  cases above, the checking procedure may be impossible to implement without significant communication overhead. A typical example is bipartiteness, whether it be applied to the network itself, or to an overlay network produced by some subroutine. 

\subsection{Proof-Labeling Schemes}

\emph{Proof-labeling scheme} (PLS)~\cite{KormanKP10} is a popular mechanisms enabling to certify correctness w.r.t.~predicates involving some global property, like bipartiteness. A PLS involves a \emph{prover} and a \emph{verifier}. The prover has access to the global state of the network (including its structure), and has unlimited computational power. It assigns \emph{certificates} to the nodes. The verifier is a distributed algorithm running at each node, performing in a single round, which consists for each node to send its certificate to its neighbors. Upon reception of the certificates of its neighbors, every node performs some local computation and outputs \textit{accept} or \textit{reject}. To be correct, a PLS for a predicate~$\Pi$ must satisfy:
 \[
 \begin{array}{c}
 \mbox{the global state of the network satisfies $\Pi$}\\
 \Updownarrow \\
 \mbox{the prover can assign certificates such that the verifier accepts at all nodes.} 
 \end{array}
 \]
 For instance, for bipartiteness, the prover assigns a color~0 or~1 to the nodes, and each node verifies that its color is 0 or~1, and is different from the color of each of its neighbors. If the network is bipartite then the prover can properly 2-color the nodes such that they all accept, and if the network is not bipartite then, for every 2-coloring of the nodes, some of them reject as this coloring cannot be proper. 
 
 The PLS certification mechanism has several desirable features. First, if the certificates are small then the verification is performed efficiently, in a single round consisting merely of an exchange of a small message between every pair of adjacent nodes. As a consequence, verification can be performed regularly and frequently without causing significant communication overhead. Second, if the network state does not satisfy the predicate, then at least one node rejects. Such a node can raise an alarm or launch a recovery procedure for allowing the system to return to a correct state,  or can stop a program running in an environment for which it was not designed. Third, the prover is an abstraction, for the certificates can be computed offline, either by the nodes themselves in a distributed manner, or by the system provider in a centralized manner. For instance, a protocol constructing an overlay network that is supposed to be bipartite, may properly 2-color the overlay for certifying its bipartiteness.  It follows from their features that PLSs are versatile  certification mechanisms that are also quite efficient whenever the certificates for legal instances are small. 
 

Many PLSs have been designed for certifying specific predicates on labeled graphs, including cycle-freeness~\cite{KormanKP10}, minimum-weight spanning tree (MST)~\cite{KormanK07}, planarity~\cite{FeuilloleyFMRRT21}, bounded genus~\cite{EsperetL22}, $H$-minor-freeness for small~$H$~\cite{BousquetFP21}, etc. In 2021, a breakthrough has been obtained, as a ``meta-theorem'' stating that a large set of properties have compact PLSs in a large class of networks (see~\cite{FeuilloleyBP22}). Namely, for every  $\MSO_2$ property\footnote{Monadic second-order logic (MSO) is the fragment of second-order logic where the second-order quantification is limited to quantification over sets. $\MSO_1$ refers to MSO on graphs with quantification over sets of vertices, whereas $\MSO_2$ refers to MSO on graphs with quantification over sets of vertices and sets of edges. }~$\Pi$, there exists a PLS for~$\Pi$ with $O(\log n)$-bit certificates for all graphs of bounded \emph{tree-depth}, where the tree-depth of a graph intuitively measures how far it is from being a star. This result has been extended to the larger class of  graphs with bounded \emph{tree-width}  (see~\cite{FraigniaudMRT22}), using certificates on $O(\log^2 n)$ bits, where the tree-width of a graph intuitively measures how far it is from being a tree. Although the class of all graphs with bounded tree-width includes many common graph families such as trees, series-parallel graphs, outerplanar graphs, etc., it does not contain families of \emph{dense} graphs.
 In this paper, we focus on the families of graphs with bounded \emph{clique-width}, which include families of dense graphs. 
 
 \subsection{Clique-Width}
 
Intuitively, the definition of clique-width is based on a ``programming language'' for constructing graphs, using only the following four instructions (see \cite{CourcelleMR00} for more details): 

 \begin{itemize}
 \setlength\itemsep{0em}
 \item Creation of a new vertex $v$ with some color $i$, denoted by $\mathsf{color}(v,i)$; 
 \item Disjoint union of two colored graphs $G$ and $H$, denoted by $G \parallel H$;
\item Joining by an edge every vertex colored $i$ to every vertex colored $j\neq i$, denoted by $i \Join j$;
\item Recolor $i$ into color $j$, denoted by $\mathsf{recolor}(i,j)$.
 \end{itemize}
 
For instance, the $n$-node clique can be constructed by creating a first node with color blue, and then repeating $n-1$ times the following: (1)~the creation of a new node, with color red, (2)~joining red and blue, and (3)~recoloring red into blue. Therefore, cliques can be constructed by using two colors only.  Similarly, trees can be constructed with three colors only. This can be proved by induction. The induction statement is that, for every tree~$T$, every vertex~$r$ of~$T$, and every two colors $c_1,c_2\in\{\mbox{blue, red, green}\}$, $T$~can be constructed with colors blue, red, and green such that $r$ is eventually colored~$c_1$, and every other vertex is colored~$c_2$. The statement is trivial for the single-node tree. Let $T$ be a tree with at least two nodes, let $r$ be one of its vertices, and let $c_1,c_2$ be two colors. Given an arbitrary neighbor~$s$ of~$r$, removing the edge $\{r,s\}$ results in two trees~$T_r$ and~$T_s$. By induction, construct~$T_r$ and $T_s$ separately so that~$r$ (resp.,~$s$) is eventually colored~$c_1$ (resp.,~$c_2$) and all the other nodes of $T_r$ and $T_s$ are colored $c_3\notin\{c_1,c_2\}$. Then form the graph $T_r \parallel T_s$, and, in this graph, join colors $c_1$ and $c_2$, and recolor $c_3$ into $c_2$. 

The clique-width of a graph~$G$, denoted by $\cw(G)$, is the smallest~$k\geq 0$ such that $G$ can be constructed by using $k$ colors. For instance, $\cw(K_n)\leq 2$ for every $n\geq 1$, and,  for every tree~$T$, $\cw(T)\leq 3$. A family of graphs has bounded clique-width if there exists $k\geq 0$ such that, for every graph~$G$ in the family, $\cw(G)\leq k$. Any graph family with bounded tree-depth or bounded tree-width has bounded clique-width~\cite{CorneilR05,NesetrilM12}. However, there are important graph families with unbounded tree-width (and therefore unbounded tree-depth) that have bounded clique-width. Typical examples \revision{(see \cite{CourcelleO00})} are cliques (i.e., complete graphs), $P_4$-free graphs (i.e., graphs excluding a \revision{path on four vertices} as an induced subgraph, a.k.a., cographs), and distance hereditary graphs (the distances in any connected induced subgraph are the same as they are in the original graph).

Many  NP-hard optimization problems can be solved efficiently by dynamic programming in the family of graphs with bounded clique-width. In fact, every $\MSO_1$ property on graphs has a linear-time algorithm for graphs of bounded clique-width~\cite{CourcelleMR00}. In this paper we show a similar form of ``meta-theorem'', regarding the size of certificates of PLS for monadic second-order properties of graphs with bounded clique-width. 

\subsection{Our Results}
 
Our main result is the following. Recall that a labeled graph is a pair $(G,\ell)$, where $G$ is a graph, and $\ell:V(G)\to\{0,1\}^\star$ is a function assigning a label to every node in~$G$. 
 
\begin{theorem}\label{theo:main}
Let $k$ be a non-negative integer, and let  $\Pi$ be an $\MSO_1$ property on node-labeled graphs with constant-size labels.
There exists a PLS certifying $\Pi$ for labeled graphs with clique-width at most~$k$, using $O(\log^2 n)$-bit certificates on $n$-node graphs. 
\end{theorem}

The same way several NP-hard problems become solvable in polynomial time in graphs of bounded clique-width, Theorem~\ref{theo:main} implies that several predicates for which every PLS has  certificates of polynomial size in arbitrary graphs have a PLS with certificates of polylogarithmic  size on graphs with bounded clique-width. This is for instance the case of non-3-colorability (which is a $\MSO_1$ predicate), for which every PLS has certificates of size $\tilde{\Omega}(n^2)$ bits in arbitrary graphs~\cite{GoosS16}. Theorem~\ref{theo:main} implies that non-3-colorability has a PLS with certificates on $O(\log^2 n)$ bits in graphs with bounded clique-width, and therefore in \revision{graphs of bounded tree-width}, cographs, distance-hereditary graphs, etc. \revision{This of course is extended to non-k-colorability, as well as other problems definable in $\MSO_1$ such as detecting whether the input graph does not contain a fixed subgraph \(H\) as a subgraph, induced subgraph, minor, etc.}

 In fact, Theorem~\ref{theo:main}  can be extended to properties including certifying solutions to maximization or minimization problems whose admissible solutions are defined by $\MSO_1$ properties. \revision{For instance maximum independent set, minimum vertex cover, minimum dominating set, etc}.

In the proof of Theorem \ref{theo:main}, we provide a  PLS that constructs a particular decomposition using at most \(k\cdot 2^{k-1}\) colors
(the clique-width of the decomposition). It is through that decomposition that the PLS certifies that the input graph satisfies \(\Pi\). 

An  application of Theorem~\ref{theo:main} is the certification of certain families of graphs. That is, given a graph family $\mathcal{F}$, designing a PLS for certifying the  membership to $\mathcal{F}$. Interestingly, there are some graph classes $\mathcal{F}$ that are expressible in \(\MSO_1\) and, at the same time, have clique-width at most~\(k\). Theorem~\ref{theo:main} provides  a PLS for certifying the membership to $\mathcal{F}$ in such cases. Indeed, the PLS first tries to build a decomposition of clique-width at most 
$k \cdot 2^{k-1} $.
If there is no such decomposition, then the input graph does not belong to 
$\mathcal{F}$. 
Otherwise, the PLS uses the decomposition to check the \(\MSO_1\) property that defines  $\mathcal{F}$.

\begin{corollary}
Let $k$ be a non-negative integer, and let  $\mathcal{F}$ be graph family expressible in $\MSO_1$ such that all graphs of the family have clique-width at most~$k$. Membership to $\mathcal{F}$ can be certified with a PLS using  $O(\log^2 n)$-bit certificates in $n$-node graphs. 
\end{corollary}

For instance, for every $k\geq 0$, the class of graphs with tree-width at most~$k$  can be certified with a PLS using  $O(\log^2 n)$-bit certificates. Indeed, ``tree-width at most~$k$'' is expressible in $\MSO_1$, and the class of graphs with tree-width at most~$k$  forms a family with clique-width at most~$3\cdot 2^{k-1}+1$~\cite{CorneilR05}.
Another interesting application is the certification of $P_4$-free graphs. Indeed, ``excluding $P_4$ as induced subgraph'' is expressible in $\MSO_1$, and $P_4$-free graphs form a family with clique-width at most~2~\cite{CourcelleO00}. It follows that $P_4$-free graphs can be certified with a PLS using  $O(\log^2 n)$-bit certificates. This is in contrast to the class of $C_4$-free graphs (\revision{i.e. graphs not containing a cycle on four vertices,} whether it be as induced subgraph or merely subgraph), which requires certificates on~$\tilde{\Omega}(\sqrt{n})$ bits~\cite{DruckerKO13}. In fact, in the case of cographs, the techniques in the proof of Theorem~\ref{theo:main} can be adapted so that to save one log-factor, as stated below. 

\begin{theorem}\label{theo:cographs}
The class of (induced) $P_4$-free graphs can be certified with a PLS using  $O(\log n)$-bit certificates  in $n$-node graphs. 
\end{theorem}
 
Note that there is a good reason for the huge gap in terms of certificate-size between $P_4$-free graphs and $C_4$-free graphs. The point is that, for any graph pattern~$H$, the class of $H$-free graphs has bounded clique-width if and only if $H$ is an induced subgraph of $P_4$ \cite{dabrowski2016clique}. Therefore, $C_4$-free graphs (as well as triangle-free graphs) do not have bounded clique-width, as opposed to $P_4$-free graphs (and $P_3$-free graphs, which are merely cliques).

\subsection{Related Work}

Proof-Labeling Schemes (PLSs) have been introduced and thoroughly studied in~\cite{KormanKP10}. Variants have been been considered in~\cite{GoosS16} and~\cite{FraigniaudKP13}, which slightly differ from PLSs: the former allows each node to transfer no only its certificates, but also its state, and the latter restricts the power of the oracle, which is bounded to produce certificates independent of the IDs assigned to the nodes. All these forms of distributed certifications have been extended in various directions, including tradeoffs between the size of the certificates and the number of rounds of the verification protocol~\cite{FeuilloleyFHPP21}, 
PLSs with computationally restricted provers~\cite{EmekGK22}, randomized PLSs~\cite{FraigniaudPP19}, quantum PLSs~\cite{FraigniaudGNP21}, PLSs rejecting at more nodes whenever the global state is ``far'' from being correct~\cite{FeuilloleyF22}, PLSs using global certificates in addition to the local ones~\cite{FeuilloleyH18}, and several hierarchies of certification mechanisms, including games between a prover and a disprover~\cite{BalliuDFO18,FeuilloleyFH21}, interactive protocols~\cite{CrescenziFP19,KolOS18,NaorPY20}, and even recently zero-knowledge distributed certification~\cite{BickKO22},  and distributed quantum interactive protocols~\cite{GallMN22}. 

All the aforementioned distributed certification mechanisms have been used for certifying a wide variety of global system states, including MST~\cite{KormanK07}, routing tables~\cite{BalliuF19}, and a plethora of (approximated) solutions to optimization problems~\cite{Censor-HillelPP17,EmekG20}. A vast literature has also been dedicated to certifying membership to graph classes, including cycle-freeness~\cite{KormanKP10}, planarity~\cite{FeuilloleyFMRRT21}, bounded genus~\cite{EsperetL22}, absence of symmetry~\cite{GoosS16}, $H$-minor-freeness for small~$H$~\cite{BousquetFP21}, etc. In 2021, a breakthrough has been obtained, as a ``meta-theorem'' stating that, for every  $\MSO_2$ property~$\Pi$, there exists a PLS for~$\Pi$ with $O(\log n)$-bit certificates for all graphs of bounded \emph{tree-depth}~\cite{FeuilloleyBP22}. This result has been extended to the larger class of  graphs with bounded \emph{tree-width}, using certificates on $O(\log^2 n)$ bits~\cite{FraigniaudMRT22}. To our knowledge, this is the largest class of graphs, and the largest class of boolean predicates on graphs for which it is known that PLSs with polylogarithmic certificates exist. 

The class of $H$-free graphs (i.e., the absence of~$H$ as a subgraph), for a given fixed graph~$H$, has attracted lot of attention in the distributed setting, mostly in the \textsf{CONGEST} model. Two main approaches have been considered. One, called distributed property testing, aims at deciding between the case where the input graph is $H$-free, and the case where the input graph is ``far'' from being $H$-free (see, e.g.,~\cite{BrakerskiP11,Censor-HillelFS19,EvenFFGLMMOORT17,FraigniaudRST16}). In this setting, the objective is to design (randomized) algorithms performing in a constant number of rounds. Such algorithms have been designed for small graphs~$H$, but it is not known whether there is a distributed algorithm for testing $K_5$-freeness in a constant number of rounds.  The other approach aims at designing algorithms deciding $H$-freeness performing in a small number of rounds. For instance, it is known that deciding $C_4$-freeness can be done in $\tilde{O}(\sqrt{n})$ rounds, and this is optimal~\cite{DruckerKO13}. The $\tilde{\Omega}(\sqrt{n})$-round lower bounds for $C_4$-freeness also holds for deciding $C_{2k}$-freeness, for every $k\geq 4$. Nevertheless, the best known algorithm performs in essentially $\tilde{O}(n^{1-\Theta(1/k^2)})$ rounds~\cite{eden2022sublinear}, even if faster algorithms exists for $k=2,3,4,5$, running in $\tilde{O}(n^{1-\Theta(1/k)})$ rounds~\cite{censorhillel_et_al:LIPIcs:2020:13111,drucker2014power}. Deciding $P_k$-freeness (as subgraph) can be done efficiently for all $k\geq 0$~\cite{FraigniaudO19}. However, this is not the case of deciding the absence of an \emph{induced}~$P_k$, and no efficient algorithms are known apart for the trivial cases $k=1,2,3$. The first non-trivial case is deciding cographs, i.e.,   $P_4$-freeness (as induced subgraph). 

The terminology \emph{meta-theorem} is used in logic to refer to a statement about a formal system proven in a language used to describe another language. In the study of graph algorithms, Courcelle's theorem~\cite{Courcelle90} is often referred to as  a meta-theorem. It says that every graph property definable in the monadic second-order logic $\MSO_2$ of graphs can be decided in linear time on graphs of bounded treewidth. This theorem was extended to clique-width, but for a smaller set of graph properties. Specifically, every graph property definable in the monadic second-order logic $\MSO_1$ of graphs  can be decided in linear-time  on graphs of bounded clique-width~\cite{CourcelleMR00}. Note that the classes of languages in $\MSO_1$ and $\MSO_2$ include languages that are NP-hard to decide (e.g., 3-colorability and Hamiltonicity, respectively). 
\revision{We remind that \(\MSO_2\) is as an extension of  \(\MSO_1\) which also allows quantification on sets of edges -- see Footnote~1 for a short description, or~\cite{CourcelleE12} for full details. Some graph properties, e.g., Hamiltonicity, are expressible in   \(\MSO_2\) but not in  \(\MSO_1\), nevertheless  \(\MSO_1\)  captures a large set of properties, including many classical NP-hard problems as explained above. Eventually, we emphasize again that, when comparing the two most famous meta-theorems, (1) \emph{$\MSO_2$ properties are decidable in linear time on bounded treewidth graphs} vs. (2) \emph{$\MSO_1$ properties are decidable in linear time on bounded clique-width graphs}, the former concerns a larger class of properties, but the latter concerns larger classes of graphs. }

%% file: model.tex

\section{Models}

In this section, we recall the main concepts used in this paper, including proof-labeling scheme, and cographs. 

\subsection{Proof-Labeling Schemes for MSO Properties}

For a fixed integral parameter $\lambda\geq 0$, we consider vertex-labeled graphs $(G,\ell)$, where $G=(V,E)$ is a connected simple $n$-node graph, and $\ell:V\to \{0,\dots,\lambda-1\}$. The label may indicate a solution to an optimization problem, e.g., a minimum dominating set ($\ell(v)=0$ or $1$ depending on whether $v$ is in the set or not), a $\lambda$-coloring, an independent set, etc. A labeling may also encode global overlay structures such as spanning trees or spanners, in bounded-degree graphs or in graphs provided by a distance-2 $k$-coloring, for $k=O(1)$. In the context of distributed computing in networks, nodes are assumed to be assigned distinct identifiers (ID) in $[1,n^c]$ for some $c\geq 1$, so that IDs can be stored on $O(\log n)$ bits. The identifier of a node~$v$ is denoted by~$\id(v)$. We denote by $N_G(v)$ the set of neighbors of node~$v$ in a graph~$G$, and we let $N_G[v]=N_G(v)\cup\{v\}$ be the closed neighborhood of~$v$. 

Given a boolean predicate~$\Pi$ on vertex-labeled graphs, a \emph{proof-labeling scheme} (PLS) for~$\Pi$ is a prover-verifier pair. The \emph{prover} is a non-trustable computationally unbounded oracle. Given a vertex-labeled graph $(G,\ell)$ with ID-assignment~$\id$, the prover assigns a \emph{certificate}~$c(v)$ to every node~$v$ of~$G$. The \emph{verifier} is a distributed algorithm running at every node~$v$ of~$G$. It performs a single round of communication consisting of sending $c(v)$ to all neighboring nodes $w\in N_G(v)$, and receiving the certificates of all neighbors. Given $\id(v)$, $\ell(v)$, and $\{c(w):w\in N_G[v]\}$, every node~$v$ outputs \emph{accept} or \emph{reject}. A PLS is correct if the following two conditions are satisfied: 
 \begin{itemize}
 \setlength\itemsep{0em}
 \item Completeness: If $(G,\ell)$ satisfies~$\Pi$ then the prover can assign certificates to the nodes such that the verifier accepts at all nodes;
\item Soundness:  If $(G,\ell)$ does not satisfy~$\Pi$ then, for every certificate assignment to the nodes by the prover, the verifier rejects in at least one node.
\end{itemize}
The main parameter measuring the quality of a PLS is the \emph{size} (i.e., number of bits) of the certificates assigned by the prover to each node of vertex-labeled graphs satisfying the predicate, and leading all nodes to accept. 

\subparagraph{MSO Predicates.} We focus on predicates expressible in $\MSO_1$. Recall that $\MSO_1$ is the fragment of monadic second-order (MSO) logic on (vertex-labeled) graphs that allows quantification on \revision{vertices} and \revision{on} sets of (labeled) vertices, and uses the adjacency predicate ($ \mathsf{adj}$). For instance non 3-colorability is in $\MSO_1$. Indeed, for every graph $G=(V,E)$, it can be expressed as: for all $A, B, C \subseteq V$, if $A\cup B\cup C=V$ and $A\cap B = A \cap C = B\cap C = \varnothing$ then 
\[
\exists (u,v) \in (A\times A) \cup (B\times B) \cup (C\times C) : (u\neq v) \land \mathsf{adj}(u,v). 
\]
We shall show that, although some $\MSO_1$ predicates, like non-3-colorability, require certificates on~$\tilde\Omega(n^2)$ bits in $n$-node graphs \revision{(see \cite{GoosS16})}, PLSs with certificates of polylogarithmic size can be designed for all $\MSO_1$ predicates in a rich class of graphs, namely all graphs with bounded clique-width. 

\subsection{Cographs and Cotrees}

We conclude this section by introducing a graph class that plays an important role in this paper. Recall that a graph is a cograph (see, e.g.,~\cite{BrandstadtLS99}) if it can be constructed by a sequence of parallel operations (disjoint union of two vertex-disjoint graphs) and join operations (connecting two vertex-disjoint graphs $G$ and $H$ by a complete bipartite graphs between $V(G)$ and $V(H)$). Therefore,  by definition, cographs have clique-width~2. In particular, cliques are cographs. 

It is known~\cite{BrandstadtLS99} that cographs capture precisely the class of induced $P_4$-free graphs. We shall show that, as opposed to $C_4$-free graphs, which require $\tilde\Omega(\sqrt{n})$-bit certificates to be certified by a PLS~\cite{DruckerKO13}\footnote{The lower bound in~\cite{DruckerKO13} is expressed for the \textsf{CONGEST} and \textsf{Broadcast Congested Clique} models, but it extends directly to PLSs since Set-Disjointness  has non-deterministic communication complexity~$\Omega(N)$ on $N$-bit inputs.}, $O(\log n)$-bit certificates are sufficient for certifying $P_4$-free graphs. This result is of  interest on its own, but proving this result will also play the role of a warmup before establishing our general result about graphs with bounded clique-width. Note that the class of $P_4$-free graphs (i.e., cographs) can be specified by an $\MSO_1$ formula. Roughly, the formula states that if there exists four vertices $v_1,v_2,v_3,v_4$ such that $\mathsf{adj}(v_i,v_{i+1})$ for $i=1,2,3$, then 
$
\mathsf{adj}(v_1,v_3) \lor \mathsf{adj}(v_1,v_4) \lor \mathsf{adj}(v_2,v_4).
$
$C_4$-freeness could be expressed in $\MSO_1$ as well. However, $P_4$-free graphs have clique-width~2 whereas $C_4$-free graphs have unbounded clique-width --- this is because there are \(2^{\Omega(n\sqrt{n})}\) different $C_4$-free graphs of size \(n\), but only \(2^{\mathcal{O}(n \log n)}\) \(n\)-vertex graphs of bounded clique-width.

Given a cograph~$G$, there is actually a canonic way of constructing~$G$ by a sequence of parallel and join operations~\cite{BrandstadtLS99}. As explained before, this construction can be described as a tree~$T$ whose leaves are the vertices of~$G$, and whose internal nodes are labeled $\parallel$ or $\Join$. This tree is called a \emph{cotree}, and will be used for our PLS. 

%% file: highlevel.tex

\section{Overview of our Techniques}
\label{se:hl}

The objective of this section is to provide the reader with a general idea of our proof-labeling scheme. Our construction bears some similarities with the approach used in~\cite{FraigniaudMRT22} for the certification of \(\MSO_2\) properties on graphs of bounded tree-width, with certificates of size \(O(\log^2 n)\) bits. However, extending this approach to a proof-labeling scheme for graphs with bounded clique-width requires to overcome several significant obstacles. We therefore start by summarizing the main tools used for 
the certification of \(\MSO_2\) properties on graphs of bounded tree-width (see Section~\ref{subsec:recalltw}), and then proceed with the description of the new tools required for extending the result to graphs of bounded clique-width, to the cost of reducing the class of certified properties from  \(\MSO_2\) to \(\MSO_1\) (see Sections~\ref{subsec:cw-et-nlcw}-\ref{subsec:summarylabelsize}). 

\subsection{Certifying $\MSO_2$ Properties in Graphs of Bounded Tree-Width}
\label{subsec:recalltw}

Recall that a tree-decomposition of a graph \(G\) is a tree \(T\) where each node \(x\) of \(T\), also called \emph{bag}, is a subset of \(V(G)\), satisfying the following three conditions: 
(1) for every vertex \(v\in V(G)\) there is a bag \(x \in V(T)\) that contains~\(v\), 
(2) for every edge \(\{u,v\} \in E(G)\), there is a bag \(x\) containing both its endpoints, and 
(3) for each vertex \(v \in V(G)\), the set of bags that contains \(v\) forms a (connected) subtree of \(T\). 
Let  \(\Pi\) be an \(\MSO_2\) property, and let $T$ be a tree-decomposition of the  graph~$G$. The proof-labeling scheme aims at providing each vertex with sufficient information for certifying the correctness of~\(T\), as well as the fact that \(G\) satisfies~\(\Pi\). To do so, the certificate of each vertex is divided into two parts, one called \emph{main messages}, and the other called \emph{auxiliary messages}.

\subparagraph{Main messages.} 

 The main message of a node \(v\) is a sequence \(\textrm{seq}_v\) representing a path of bags in \(T\) that connects a leaf with the root, such that \(v\) is contained in at least one bag of \(\textrm{seq}_v\).  For each bag \(x \in \textrm{seq}_v\), the main message includes, roughly: the set of vertices contained in \(x\), the identifier of a vertex \(\ell_x\) in \(x\), called the \emph{leader} of \(x\), and a data structure \(c_x\) used to verify the \(\MSO_2\) property $\Pi$ on~\(G\). The leader \(\ell_x\) of \(x\) is chosen arbitrarily among the vertices of \(x\) that are adjacent to a vertex \(u\) belonging to the parent bag \(p(x)\) of \(x\) in \(T\). The vertex \(u\) is said to be \emph{responsible} for \(x\) in \(p(x)\). Let us assume the following consistency condition: for every bag \(x\) of \(T\),  every vertex in \(x\) received the same information about all the bags from \(x\) to the root of~\(T\). Under the promise that the consistency condition holds, it is possible to show that the vertices can collectively verify that \(T\) is indeed a tree-decomposition of~$G$, and that \(G\) satisfies~\(\Pi\).  

\subparagraph{Auxiliary messages.} 

The role of the auxiliary messages is precisely to check the above consistency condition.  For each bag \(x\), let  \(\tau_x\) be a Steiner tree \revision{(i.e. a minimal tree connecting a set of vertices denoted \emph{terminals})}  in $G$ rooted at the leader~\(\ell_x\), with all the nodes of \(x\) as terminals. Every vertex in \(\tau_x\) receives an auxiliary message containing the certification of~$\tau_x$ \revision{(each vertex of \(\tau_x\) receives the identifier of a root, of its parent and the distance to the root)}, and a copy of the information about \(x\) given to the nodes in the bag~$x$, through their main messages. By using the auxiliary messages, the leader $\ell_x$ can verify whether the subgraph $G[x]$ of \(G\) induced by the union of the bags in \(T[x]\) satisfies~$\Pi$, where \revision{\(T[x]\) the subtree of \(T\) containing \(x\) and all its descendants}. Specifically, this verification is performed by simulating the dynamic programming algorithm in Courcelle's Theorem~\cite{Courcelle90} as in the version of Boire, Parker and Tovey \cite{BoPaTo92}. This uses a \revision{constant-size} data structure \(c_x\) stored in the auxiliary messages that ``encodes'' the predicate \(\Pi(G[x])\). Its correctness can be verified by a composition of the values~\(c_y\) for each child \(y\) of \(x\) in~$T$. The tree \(\tau_x\) is actually used to transfer the information about \(c_y\) from the node $\ell_y$ in \(x\) responsible for~\(y\), to the leader \(\ell_x\).  

\subparagraph{Certificate size.} 

If \(T\) is of depth \(d\), then the main messages are of size \(O(d \log n)\) bits.  Crucially, for every graph \(G\), there is a tree-decomposition \(T\) satisfying that, for every bag \(x\), there is a Steiner tree \(\tau_x\) completely contained in \(G[x]\). Such a decomposition is called \emph{coherent} in~\cite{FraigniaudMRT22} (Lemma 3). It follows that every node participates in a Steiner tree with at most \(d\) bags, which implies that the auxiliary messages can be encoded in \(O(d \log n)\) bits. Thanks to a construction by Bodlaender~\cite{bodlaender1989nc}, it is possible to choose a coherent tree-decomposition with depth \(d = O(\log n)\), up to increasing the sizes of the bags by a constant factor only. It follows that the certificates are of size $O(\log^2n)$ bits. 
 
\bigbreak 

Our construction also follows the general structure described above. However, each element of this construction has to be adapted in a highly non-trivial way. Indeed, the grammar of clique-with, and the related structure of NLC decomposition, differ in several significant ways from the grammar of tree-width. The rest of the section is dedicated to providing the reader with a rough idea of how this can be done. 

\subsection{Clique-Width and NLC-Width}
\label{subsec:cw-et-nlcw}
 
 First, instead of working with clique-width, it is actually more convenient  to work with the NLC-width, where NLC stands for \emph{node-label controlled}. Every graph of clique-width at most~$k$ has NLC-width at most~$k$, and every graph of NLC-width at most~$k$ has clique-width at most~$2k$~\cite{Johansson98}. As clique-width, NLC-width can be viewed as the following grammar for constructing graphs, bearing similarities with the grammar for clique-width:
 \begin{itemize}
 \setlength\itemsep{0em}
 \item Creation of a new vertex $v$ with color $i\in\mathbb{N}$, denoted by $\mathsf{newVertex}_i$;
\item Given a set $S$ of ordered pairs of colors, and an ordered pair $(G,H)$ of vertex-disjoint colored graphs, create a new graph as the union of $G$ and $H$, then join by an edge every vertex colored~$i$ of~$G$ to every vertex colored~$j$ of~$H$, for all $(i,j)\in S$; this operation is denoted by $G \Join_S H$;
\item Recolor the graph, denoted by $\mathsf{recolor}_R$ where $R:\mathbb{N}\to\mathbb{N}$ is any function.

\end{itemize}
If $k\geq 1$ colors are used, a recoloring function~$R$ is a function $R:[k]\to [k]$. When $R$ is used,  for every $i\in [k]$, vertices with color~$i$ are recolored~$R(i)\in [k]$ (all colors are treated simultaneously, in parallel). Note that the recoloring operation in the definition of clique-width is limited to functions~$R$ that preserve all colors but one. Note also that, for $S=\varnothing$, the operation  $G \Join_S H$ is merely the same as  $G \parallel H$ for clique-width. We therefore use $G \Join_\varnothing H$ or $G \parallel H$ indistinctly. The NLC-width of a graph~$G$ is the smallest number of colors such that $G$ can be constructed using the  operations above. It is denoted by $\nlcw(G)$. For instance, the $n$-node clique can be constructed by creating a first node $v_1$ with color~1, and then repeating, for all $i=1,\dots,n-1$, (1)~the creation of a new node~$v_{i+1}$, with color~1 as well, and (2)~applying $v_{i+1} \Join_{\{(1,1)\}}K_i$ to get the clique~$K_{i+1}$ on $i+1$ vertices. Therefore, cliques can be constructed by using one color only, i.e., $\nlcw(K_n)=1$ for every $n\geq 1$. 

\subparagraph{NLC-decomposition.} For every $k\geq 1$,  the construction of a graph~$G$ with $\nlcw(G)\leq k$ can be described by a binary tree~$T$, whose leaves are the (colored) vertices of~$G$. In~$T$, every internal node~$x$ has an identified  left child~$x'$ and an identified right child~$x''$, and is labeled by~$\parallel$ or~$\Join_S$ for some non-empty set $S\subseteq [k]\times [k]$. This label indicates the operation performed on the (left) graph $G'$ with vertex-set equal to the leaves of the subtree~$T_{x'}$ of~$T$ rooted at~$x'$, and the (right) graph $G''$ with vertex-set equal to the leaves of the subtree~$T_{x''}$ of~$T$ rooted at~$x''$. That is, node $x$ corresponds to the operation $G_{x'}\parallel G_{x''}$ or $G_{x'}\Join_S G_{x''}$, depending on the label of~$x$.  In addition to its label ($\parallel$~or~$\Join_S$ for some~$S\neq\varnothing$), a node may possibly also include a recoloring function~$R:[k]\to [k]$, which indicates a recoloring to be performed \emph{after} the join operation, see Figure~\ref{fig:NLCdec} for an example.

\begin{figure}[ht]
\begin{center}
\includegraphics[scale=0.2]{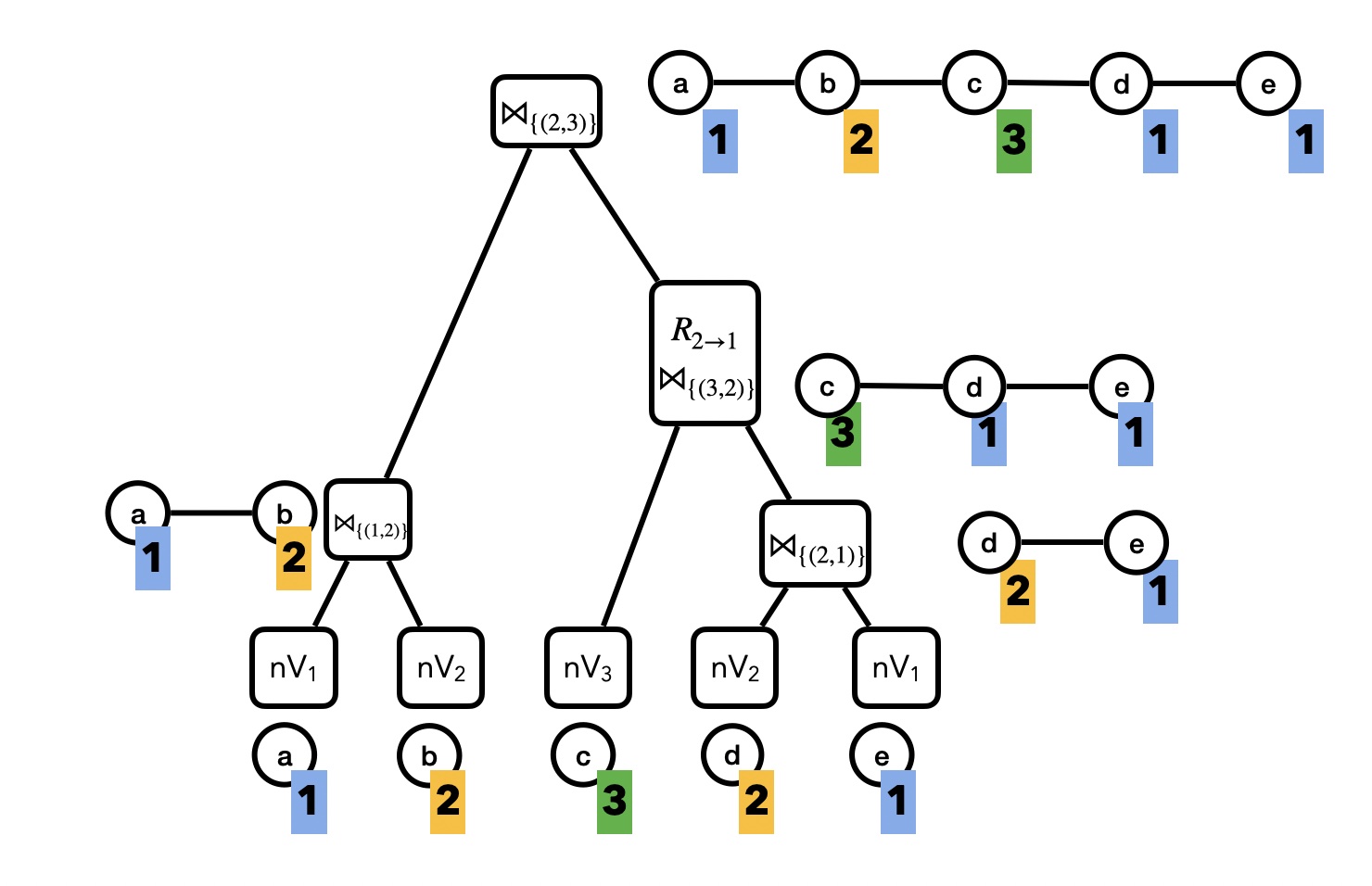}
\caption{An NLC decomposition tree $T$. Next to each  node $x$ of the tree is displayed the colored graph $G[x]$ corresponding to subtree $T[x]$ of $T$ rooted at~$x$.}
\label{fig:NLCdec}
\end{center}
\end{figure}

\subsection{From Tree-Width to NLC-Width: The Main Messages}

Let \(\Pi\) be an \(\MSO_1\) property, and let $T$ be an NLC-decomposition tree of a graph \(G\) with \(\cw(G)\leq k\). That is, we can choose the tree \(T\) as one using at most \(k\) colors. In the following, to avoid confusion, we call \emph{vertices} the elements of the vertex set of~\(G\), and \emph{nodes} the elements of the vertex-set of the decomposition tree~\(T\). The structure of our certificates differ from the one in~\cite{FraigniaudMRT22}, and now we decompose the certificate assigned to each node~$v$ into three parts: \emph{main messages}, \emph{auxiliary messages}, and \emph{service messages}. This subsection focuses on the main messages. 

Our main messages have, to some extent, a structure similar to the main messages used in~\cite{FraigniaudMRT22}  for the tree-width.  In particular, vertex \(v\) receives a sequence \(\textsf{path}(v)\), listing all the nodes, i.e., the whole set of operations, in the path from the root of  \(T\)  to the leaf of $T$ where \(v\) was created. For each node \(x\) in $\textsf{path}(v)$, the main message also includes the vertex identifier of a leader for~$x$, called \emph{exit vertex of \(x\)}, and denoted by \(\textsf{exit}(x)\). The main message also includes a data structure \(h(x)\) that encodes the truth value of the \(\MSO_1\) property on \(G[x]\). However, unlike the case of tree-width, where the nodes of the tree-decomposition are sets of vertices (i.e., bags) of bounded size, the contents of a non-leaf node in an NLC-decomposition tree \(T\) does not necessarily include information about the vertices created in \(T[x]\). For that reason, our proof-labeling scheme includes additional information in the main message of \(v\) in order to verify the correctness of the given decomposition. It may actually be worth providing a concrete example to explain the need for additional information. 

\subparagraph{Example.} 

Let us pick an arbitrary edge \(\{u,v\} \in E(G)\), and denote by \((x_1(u), \dots, x_{t_1}(u))\) and \((x_1(v), \dots, x_{t_2}(v))\)  the sequences \(\textsf{path}(u)\) and \(\textsf{path}(v)\), respectively, where \(x_1(u) = x_1(v)\) is the root of \(T\), and \(x_{t_1}(u)\) and \(x_{t_2}(v)\) are the nodes where \(u\) and \(v\) are respectively created.  Let \(x_1, \dots, x_{t_3}\) be the longest common prefix of these two sequences, i.e., the information contained in their main messages coincide on the first  \(t_3\) elements,  but \(x_{t_3+1}(u) \neq x_{t_3+1}(v)\). In the tree~$T$, \(x_{t_3+1}(u)\) and \(x_{t_3+1}(v)\) are two children of \(x_{t_3}\). The sequence of operations described in \(x_{t_1}(u), \dots, x_{t_3+1}(u)\) defines the color \(c(u)\) that \(u\) has in \(x_{t_3+1}(u)\). Similarly, the color \(c(v)\) of \(v\) in \(x_{t_3+1}(v)\) is defined by \(x_{t_2}(v), \dots, x_{t_3+1}(v)\). In order to create the edge \(\{u,v\}\), the operations described in \(x_{t_3}\) must specify a \(\Join\) operation between vertices with color~\(c(u)\) and vertices with color~\(c(v)\). However, the join operations described in an NLC-decomposition tree make a clear distinction between the left and right children of a  node. Therefore, in our example, for checking that the edge \(\{u,v\}\) is indeed correctly created in the given decomposition tree, the vertices \(u\) and \(v\) must be able to distinguish which of the two children \(x_{t_3+1}(u)\) and \(x_{t_3+1}(v)\) of \(x_{t_3}\) is the left child, and which one is the right child. 

\bigbreak

For a node \(x\) different from the root, let us denote by \(p(x)\) the parent of \(x\) in \(T\). The main message of~\(v\) includes 
a sequence \(\textsf{links}(v)\) that specifies, for each node \(x\) in \(\textsf{path}(v)\) different from the root, whether $x$ is the left or right child of \(p(x)\). For instance, in the example of Figure~\ref{fig:NLCdec}, we have  \(\textsf{links}(c) = (1,0)\), indicating that, to reach the leaf creating vertex $c$ from the root, one must follow the right child~(1), and then the left child~(0). S imilarly, \(\textsf{links}(d)=(1,1,0)\).  The sequences \(\textsf{links}\) are
also used to determine the longest common prefixes of the main messages, when the  same operations are repeated between two children of a same node (consider for instance the case where the same operation is performed at all the nodes of the decomposition tree). Back to our example above, let us suppose that the sequences \(\textsf{links}(u)\) and \(\textsf{links}(v)\) specify that \(x_{t_3+1}(u)\) is the left neighbor of \(x_{t_3}\), and \(x_{t_3+1}(v)\) is the right neighbor of \(x_{t_3}\). Using this information, \(u\) and \(v\) can  infer that it is an operation \(\Join_S\), with \((c(u), c(v)) \in S\) that is specified in the description of \(x_{t_3}\). 
With the given information, each vertex can thus check that all its incident edges are indeed created at some node of the decomposition tree~$T$. 

It remains to check that the decomposition does not define non-existent edges. To do so, the main message of every vertex \(v\) also includes, for each node \(x\) in \(\textsf{path}(v)\), and for each \(i \in [k]\),  the integers \(\textsf{color}_i(x)\) representing the number of vertices of \(G[x]\) that are colored \(i\) in the root of \(T[x]\). (Recall that the subgraph $G[x]$ is the subgraph of $G$ induced by the vertices created in the subtree \(T[x]\) of~$T$). Returning to our example, vertex \(v\) checks that it has exactly \(\textsf{color}_{c(u)}(x_{t_3+1}(u))\) neighbors with the same longest common prefix as \(u\) colored \(c(u)\) in the left children of \(x_{t_3}\). Also, vertex \(v\) checks, for each \(i\in [k]\), that the number of vertices colored~\(i\) in node \(x_{t_3}\) corresponds to the sum of the number of vertices colored \(j\) in \(x_{t_3+1}(u)\) and \(x_{t_3+1}(v)\), for each color \(j\) that is recolored \(i\) by the recoloring operation defined in \(x_{t_3}\). So, let us assume that the following consistency condition (analogous to the one for the certification of tree-decompositions) holds: 

\begin{description}
\item[C1:] For every pair of vertices \(u,v \in V(G)\), and for every  node \(x\) in both \(\textsf{path}(u)\) and \(\textsf{path}(v)\),  \(u\)~and \(v\) receive the same information about all nodes in the path from \(x\) to the root of \(T\) in their main messages, and 
\item[C2:]  If \(x\) is the root of \(T\), then the data structure \(h(x)\) describes an accepting instance (i.e., \(G\) satisfies \(\Pi\)). 
\end{description}

Assuming that the consistency condition is satisfied, it is not difficult to show that the vertices can collectively check that the given certificates indeed represent an NLC-decomposition tree, and that \(G\) satisfies \(\Pi\).  The difficulty is however in checking that the consistency condition holds. This is the role of the auxiliary and service messages, described next. 

\subsection{Checking Consistency: Auxiliary, and Service Messages }

We use auxiliary and service messages for allowing our proof-labeling scheme to check the first condition~\textbf{C1} of the consistency condition defined at the end of the previous subsection. 

Auxiliary messages can easily be defined for every node \(x\) of \(T\) satisfying that  \(G[x]\) is connected. In that case, the auxiliary messages of all the vertices \(v\) in \(T[x]\) contain the certificates for certifying a spanning tree \(\tau_x\) of \(G[x]\) rooted at the exit vertex of \(x\). Each vertex \(v\) can verify that the longest common prefix common to $v$ and its parent in \(\tau_x\) contains all the nodes from the root up to \(x\), and that the information given in the main messages coincide for all such nodes. Observe that every vertex \(v\) may potentially contain one auxiliary message for each node in  \(\textsf{path}(v)\). 

 The case where \(G[x]\) is not connected is fairly more complicated, and we need to introduce another type of decomposition. 

\subparagraph{NLC+ decompositions trees.}

 Observe that \(G\) itself is connected. Therefore, there must exist an ancestor \(z\) of \(x\) for which  \(G[z]\) is connected. We could provide the vertices in $G[z]$ with a spanning tree of \(G[z]\) for checking the consistency in \(T[x]\). However, the vertices in \(G[z]\) do not necessarily contain \(x\) in the prefixes of their node sequences, so we would have to put a copy of the main message associated to \(x\) on every node participating in the spanning tree. Since an NLC-decomposition tree does not allow to provide a bound on the distance between \(z\) and \(x\) in the tree, we have no control on how many copies of main messages a vertex should handle.
 
  Therefore, to cope with the case where \(G[x]\) is disconnected, we define a specific type of NLC decompositions trees, called NLC+ decompositions trees. The NLC+ decomposition trees are similar to NLC-decomposition trees, up to two important differences. 
\begin{itemize}
\item First, we allow the nodes corresponding to a \(\parallel\) operation to have arbitrary large arity, and thus NLC+ decomposition trees are not binary trees, as opposed to NLC-decomposition trees. 
\item Second, if a node \(x\) induces a disconnected subgraph \(G[x]\), then its parent node \(p(x)\) must satisfy that \(G[p(x)]\) is connected. Observe that  \(p(x)\) must then correspond to a \(\Join\) operation, and thus $p(x)$ has only two children: $x$ and another child, denoted by $y$. 
\end{itemize}

\subparagraph{Service trees.}

 A \emph{service tree} \(S_x\) for a node \(x\) such that $G[x]$ is disconnected is a Steiner tree in \(G[p(x)]\) rooted at the exit vertex of \(x\), and with all the vertices of \(G[x]\) as terminals. Each vertex of  \(S_x\) (i.e., all vertices in \(G[x]\), plus some vertices in \(G[y]\) is given a \emph{service message}, which contains the certificate for the tree \(S_x\), as well as a copy of the information about \(x\) given in the main messages of the vertices in \(G[x]\). Each vertex in \(S_x\) can then check that it shares the same information about \(x\) than its parent. The properties of the NLC+ decomposition guarantee that a vertex \(v\) participates to at most two service trees, for each node \(x\) in the sequence \(\textsf{path}(v)\). Indeed, vertex \(v\) necessarily participates in \(S_x\) when \(x\) is of type \(\parallel\),  and may also participate in \(S_y\) whenever  the sibling $y$ of \(x\) is of type \(\parallel\). There are significantly more subtle details concerning service trees, but they will be described in Section~\ref{se:cwd}.

It remains to check the second condition~\textbf{C2} of the consistency condition defined at the end of the previous subsection, which consists in verifying the correctness of \(h(x)\), for every node \(x\) of~\(T\). This is explained next. 

\subsection{Dealing with $\MSO_1$ Predicates}

In their seminal work, Courcelle, Makowsky and Rotics~\cite{CourcelleMR00} proved that every $\MSO_1$ predicate $\Pi$ on vertex-labeled graphs can be decided in linear time on graphs of bounded clique-width, and hence on graphs of bounded NLC-width, whenever a decomposition tree is part of the input. The running time of the algorithm is~$O(n)$, i.e., linear in the number~$n$ of vertices of the input graph, with constants hidden in the big-O notation that depend on the clique-width bound, on the number of labels, and on the $\MSO_1$ formula encoding the predicate~$\Pi$.  Note that this result does not hold for $\MSO_2$ predicates, which is why our proof-labeling scheme applies to $\MSO_1$ predicates only. We discuss the possible extension to $\MSO_2$ properties in the conclusion (see Section~\ref{sec:conclusion}).

For our purpose it is convenient to see the linear-time decision algorithm  as a dynamic programming algorithm over the NLC-decomposition tree of the input graph. We formalize this dynamic programming approach following the vocabulary and notations used by Borie, Parker and Tovey~\cite{BoPaTo92}. Note that the latter  provided an alternative proof of Courcelle's theorem, but for graphs of bounded tree-width, i.e., specific to a graph grammar defining tree-width.  To design our proof-labeling scheme, we adapt their approach to a graph grammar defining  NLC-width.

\subparagraph{Homomorphism Classes.}

For a fixed property $\Pi$ and a fixed parameter $k$, there is a finite set $\mathcal{C}$ of \emph{homomorphism classes} (whose size depends only on $\Pi$ and $k$) such that we can associate to each graph  $G$ of clique-width at most \(k\) its class $h(G) \in \mathcal{C}$ \revision{(for more details see Proposition \ref{pr:reg})}. Whenever $G$ is obtained from two graphs $G_1$ and $G_2$ by a $\Join_S$ operation potentially followed by a recoloring operation~$R$, the class $h(G)$ only depends on $h(G_1)$, $h(G_2)$, $S \subseteq [k] \times [k]$, and $R:[k] \to [k]$. This property also holds whenever $\Join_S$ is replaced by $\parallel$. Moreover, we also extend the notion to arbitrary arity so that it holds for the NLC+ decomposition trees. Importantly, Courcelle's theorem~\cite{CourcelleMR00} provides a "compiler" allowing to compute $h(G)$ whenever $G$ is formed by a single vertex of color $j \in [k]$, and to compute $h(G)$ from $h(G_1)$, $h(G_2)$, $S$ and $R$ whenever $G = R(G_1 \Join_S G_2)$. 

\subparagraph{Checking Condition \textbf{C2}.}

In our proof-labeling scheme, for each node $x$ of the NLC+ decomposition tree, we specify \(h(x)\) as the class $h(G[x])$. Following the same principles as before, the consistency of these classes can be checked by simulating a bottom-up parsing of the decomposition tree, in a way very similar to what we described before for checking the consistency of $\textsf{color}(x)$, but replacing the mere additions by updates of the homomorphism classes as described above. 

This completes the rough description of our proof-labeling scheme. 

\subsection{Certificate Size}
\label{subsec:summarylabelsize}

For each vertex \(v\), the main, auxiliary, and service messages of \(v\) can be encoded using \(O(\log n)\) bits for each node \(x\) in \(\textsf{path}(v)\), for the following reasons. 
\begin{itemize}
\item The main message associated to a node \(x\) contains the following information. First, the list of operations described in the node, which can be encoded in \(O(k^2)\) bits.  Second, the corresponding index of \(\textsf{links}\), which is just one bit representing whether \(x\) is the left or right children of its parent. Third, the homomorphism class \(h(x)\) that can be encoded in \(f(k)\) bits for some function~$f$ depending on the $\MSO_1$ property under consideration --- see the remark further in the text for a discussion about~$f$. Finally, it includes the node identifier of the exit vertex of \(x\), and the integers \(\textsf{color}_i\) for each \(i\in [k]\).  All these latter items can be encoded on  \(O(\log n)\) bits. 

\item The auxiliary message associated to node \(x\) (whenever \(G[x]\) is connected) corresponds to the certification of a spanning tree of \(G[x]\), which can be encoded in \(O(\log n)\) bits (see~\cite{KormanKP10}). 

\item For the service messages, note that vertex \(v\) participates in at most two service trees associated to \(x\): the one of \(x\) (whenever \(G[x]\) is disconnected), plus the one of the sibling \(y\) of \(x\) (when \(G[y]\) is disconnected). Again, each of these trees can be certified using \(O(\log n)\) bits.
\end{itemize}

Therefore, the total size of the certificates is \(O(d\cdot \log n)\) bits, where \(d\) is the depth of the NLC+ decomposition tree~\(T\). Our final certificate size depends then on how much we can bound the depth~$d$ of \(T\). Courcelle and Kant\'e \cite{courcelle2007graph} show that there always exists an NLC decomposition tree of logarithmic depth, but it comes with a price: the width of the small depth decomposition can be exponentially larger than the width of the original decomposition. Specifically, Courcelle and Kant\'e have shown that every \(n\)-node graph of NLC-width \(k\) admits an NLC-decomposition of width \(k\cdot 2^{k+1}\) such that the corresponding decomposition tree \(T\) has depth \(\mathcal{O}(\log n)\). Fortunately, our construction of NLC+ decomposition trees does not increase the depth of a given NLC-decomposition tree. In other words, we can use the result of Courcelle and Kant\'e to also show that  NLC+ decomposition trees have logarithmic depth. Overall, we conclude that the certificate size is \(O(\log^2 n)\) bits. 

\subparagraph{Remark.}

Our asymptotic bound on the size of the certificates hides a large dependency on  the clique-width \(k\) of the input graph. For certifying the NLC+ decomposition only, the constant hidden in the big-O notation is single-exponential in $k$, given that the width of the NLC+ decomposition tree with logarithmic depth grows to \(k\cdot 2^{k+1}\). However, for certifying an \(\MSO_1\) property, the dependency on \(k\) can be much larger, as it depends on the number of homomorphism classes.  It is known that, for  $\MSO_1$ properties, the number of homomorphism classes is at most a tower of exponentials in $k$, where the height of the tower depends on the number of quantifiers in the $\MSO_1$ formula. Moreover, this non-elementary dependency on $k$ can not be improved significantly~\cite{FrickGrohe04}. This exponential or even super-exponential dependency on the clique-width~$k$ is however inherent to the theory of algorithms for graphs of bounded clique-width. The same type of phenomenon
 occurs when dealing with graphs of 
 bounded tree-width (see~\cite{FrickGrohe04}), and the proof-labeling scheme in~\cite{FraigniaudMRT22} is actually subject to the same type of dependencies in the bound~$k$. On the other hand, the certificate size of our proof-labeling scheme grows only polylogarithmically with the size of the graphs.

%% file: cograph.tex

\section{Certifying Cographs}\label{se:cographs}

In this section we describe a PLS for the recognition of cographs, using $O(\log n)$-bit certificates in $n$-node graphs. That is, this section is entirely dedicated to  proving Theorem~\ref{theo:cographs}. Our scheme uses two known technical lemmas. The first lemma (Theorem~4 in~\cite{courcelle2007graph}) states that every graph of NLC-width \(k\) admits an NLC-decomposition of \emph{logarithmic} depth, and width still bounded by a function of~$k$. 

\begin{lemma}[Courcelle and Kant\'e \cite{courcelle2007graph}]\label{lem:logdepth}
Every \(n\)-node graph of NLC-width \(k\) admits an NLC-decomposition of width \(k\cdot 2^{k+1}\) such that the corresponding decomposition tree \(T\) has depth \(\mathcal{O}(\log n)\).
\end{lemma} 

The second technical lemma used to establish Theorem~\ref{theo:cographs} states that every cograph has a spanning tree with very small diameter. 

\begin{lemma}[Montealegre, Ram\'{i}rez-Romero, and Rapaport \cite{montealegre2021compact}]\label{lem:depth2tree}
Every cograph has a rooted spanning tree of depth~2 in which every node at depth~1 in the tree has at most one child.
\end{lemma}

Our proof  of Theorem~\ref{theo:cographs} is structured as follows. First, we describe the certificates assigned by the prover at each node. Next, we describe the verification algorithm, and we prove that the scheme satisfies soundness and completeness. Finally, we establish the desired upper bound on the size of the certificates.

\subsection{Certificate Assignment }  

Let $G=(V,E)$ be the considered graph, and let \(u\in V\).  The certificate $c(u)$ is divided in two parts, respectively called \emph{main message} and \emph{auxiliary message}.

\subparagraph{Main messages.} 

Lemma~\ref{lem:logdepth} states that there exists an NLC-decomposition of cographs, with width~4, and such that the corresponding decomposition tree \(T_{dec}\) has depth \(\mathcal{O}(\log n)\). These main messages are used to encode such an NLC-decomposition tree \(T_{dec}\).  At every node~$u$, the main message contains the following data:
\begin{itemize}
\item The identifier $\id(u)$, and an integer \(\textsf{deg}(u)\) representing the degree of \(u\) in $G$.

\item A sequence \(\textsf{path}(u) = (x_1(u), \dots, x_{d}(u))\) of values, representing a path in \(T_{dec}\) from the root of $T_{dec}$ to the leaf of $T_{dec}$ where \(u\) is created thanks to $\mathsf{newVertex}$. Here \(d = d(u)\) represents the length of \(\textsf{path}(u)\). For each \(i \in \{1, \dots, d\}\), the value \(x_i(u)\) is a list of all the operations (type of join, potential recoloring, etc.) performed at the $i$th node of the path, starting from the root. For simplicity, we also refer to  \(x_i(u)\)  as this $i$th node. 

\item A sequence \(\textsf{links}(u) = (\ell_1(u), \dots, \ell_{d}(u)) \in \{0,1\}^{d}\), representing the sequence of edges that are followed to reach the $i$th node \(x_i(u)\) of the path $\textsf{path}(u)$ from the root $x_1(u)$. More precisely, \(\ell_1(u) = 0\) and for each \( i \in \{2, \dots, d\}\),  \(\ell_i(u)= 0\) if \(x_{i}(u)\) is the left child of \(x_{i-1}(u)\), and \(\ell_i(u) = 1\) otherwise. 
\end{itemize}

Note that since the prover provides \(u\) with the whole list of  \(\mathcal{O}(\log n)\) operations from the node of $T_{dec}$ where $u$ is created  to the root of~$T_{dec}$, there is not enough space for assigning a unique identifier to each nodes of the tree~$T_{dec}$, as this would results in consuming $O(\log^2n)$ bits in $\textsf{path}(u)$. Instead, a node of the decomposition will be uniquely identified by the sequence \(\textsf{links}(u)\) and by the content of the values stored in \(\textsf{path}(u)\). We shall show that this is sufficient.

\subparagraph{Auxiliary messages.} 

Lemma~\ref{lem:depth2tree} states that there is a rooted spanning tree $T_{span}$ of $G$, with depth~2, and in which every node at depth~1 in the tree has at most one child. The auxiliary messages are used to gather all the main messages on a single node \(r \in V\). Each node receives the information required to certify a depth-2 spanning tree $T_{span}$ rooted at $r$. In addition, every node  at depth~1 in $T_{span}$ receives the main messages of its child in~$T_{span}$. Formally, the auxiliary message provided by the prover to a node \(u\) contains the following data:

\begin{itemize}
\item The identifier $\rho(u)=\id(r)$ of the root $r$ of $T_{span}$.
\item An integer \(\textsf{depth}(u) \in \{0,1,2\}\) representing the depth of \(u\) in $T_{span}$.
\item If \(\textsf{depth}(u) = 2\), a node identifier \(\textsf{parent}(u)\) representing the parent of \(u\) in $T_{span}$.
\item If \(\textsf{depth}(u) = 1\), a variable \(\textsf{child}(u)\), either representing the node identifier of the child of \(u\) in $T_{span}$, or \(\bot\) if \(u\) has no children in $T_{span}$.
 \item If \(\textsf{depth}(u) = 1\), and \(\textsf{child}(u) \neq \bot\), then \(u\) receives a variable \(\textsf{M}(u)\) representing  the main message of \(\textsf{child}(u)\).
\end{itemize}

\subsection{Verification Scheme} 

We now describe the verification algorithm performed by every vertex \(u\) of the actual graph. First,  \(u\) verifies that \(\textsf{deg}(u)\) corresponds to its degree, and that all the values stored in \(\textsf{path}(u)= (x_1(u), \dots, x_{d}(u))\) effectively correspond to a list of operations of an NLC-decomposition of width~4. In particular, $u$ checks that  \(x_{d}(u)\) contains the unique operation  \(\textsf{NewVertex}_i\) for some \(i\in \{1,\dots,4\}\). Concretely, after having shared its certificate with its neighbors,   \(u\) checks the following conditions, for each \(v \in N(u)\):

\begin{enumerate}
\item \label{item:pierre-1} \(\textsf{links}(u)\) and \(\textsf{links}(v) \) have a common prefix. More precisely, \(u\) checks that there exists an index \(i \in \{1, \dots, d(u)\}\)  such that \(\ell_j(u) = \ell_j(v)\) for every \(j \leq i\). Over all such indices  we denote by \(i^*\) the maximum one.

\item \label{item:pierre-2} \(\textsf{path}(u)\) and \(\textsf{path}(\revision{v})\) share the same the first \(i^*\) coordinates. More precisely, \(u\) verifies that \(x_j(u) = x_j(v)\) for every \(j \leq i^*\).

\item \label{item:pierre-3} Let \(\mathsf{currentcolor}(u,i^*)\) be the color of \(u\)  resulting from the \(\mathsf{NewVertex}\) operation specified in \(x_{d(u)}\), and all the \(\mathsf{recolor}\) operations in all the nodes in \(\textsf{path}(u)\) up to \(x_{i^*}(u)\), but not including the \(\mathsf{recolor}\) operations in \(x_{i^*}\)). The value \(\mathsf{currentcolor}(v,i^*)\) is defined the same for node~$v$. The following holds:
\begin{itemize}
\item  If \(\ell_{i^*}(u) = 0\) then \(u\) checks that  the join operation $\Join_S$ in  \(x_{i^*}\) satisfies  
\[(\mathsf{currentcolor}(u,i^*), \mathsf{currentcolor}(v,i^*)) \in S.\]
 \item If \(\ell_{i^*}(u) \neq 0\) then \(u\) checks that  the join operation $\Join_S$ in  \(x_{i^*}\) satisfies  
 \[(\mathsf{currentcolor}(v,i^*),\mathsf{currentcolor}(u,i^*)) \in S.\]
 \end{itemize}
\item \label{item:agree-root} Nodes $u$ and $v$ agree on the root of $T_{span}$, i.e., \(\rho(u) = \rho(v)\).  


 \item \label{item:there-exists-parent} If \(\textsf{depth}(u)=2\) then \(u\) checks that there exists \(v \in N(u)\) such that \({\textsf{parent}(u) = v}\). 
 
\item \label{item:neighbor-root} If \(\textsf{depth}(u)=1\), then \(u\) checks that 
\begin{itemize}
\item \(\rho(u) \in \{\id(v):v\in N(u)\}\);
\item If  \(\textsf{child}(u)  \neq \bot\) then (1)~\(\textsf{child}(u) \in N(u)\),  (2)~\(\textsf{depth}(\textsf{child}(u)) = 2\), and (3)~every \(v \in N(u)\smallsetminus \{\textsf{child}(u) \} \) with  \(\textsf{depth}(v) = 2\) satisfies that \(\textsf{parent}(v) \neq u\);
\item If  \(\textsf{child}(u)  \neq \bot\) then \(\textsf{M}(u)\) equals the main message of \(\textsf{child}(u)\).
\end{itemize}

\item  \label{item:number-9-number-9} If \(\id(u) = \rho(u)\) then, \(u\) checks that the following holds\footnote{If all previous conditions are satisfied, then the root \(r=u\) obtains from its neighbors all the main messages of the nodes in~\(G\)}:
\begin{itemize}
\item The information in \(\{\textsf{path}(v): v\in V\}\) and \(\{\textsf{links}(v):v\in V\}\) is consistent, that is, for every \(v_1, v_2 \in V\), and for  every \(i \in \mathbb{N}\),
\[ \Big(\forall j \leq i, \; \ell_j(v_1) = \ell_j(v_2) \Big) \Longrightarrow \Big (\forall j \leq i,  \; x_j(v_1) = x_j(v_2)\Big).\]
Observe that if this condition is satisfied, then necessarily \(\{\textsf{path}(v): v\in V\}\) and \(\{\textsf{links}(v):v\in V\}\) describe a unique NLC-decomposition tree. Let us denote this tree by~\(T(u)\), and let \(G^*\) be the graph corresponding to the realization of the NLC-decomposition tree given by \(T(u)\). Note that \(u\) can obtain all the vertices and edges of \(G^*\) from \(T(u)\). 
\item Node \(u\) checks that \(G^*\) is a cograph.
\item Finally \(u\) checks that, for every node \(v \in V\),  the number of neighbors of \(v\) in \(G^*\)  equals \(\textsf{deg}(v)\). 
\end{itemize}
\end{enumerate}

\subsection{Completeness and Soundness} 

The completeness of our scheme directly follows from the fact that, thanks to Lemma~\ref{lem:logdepth}, every cograph admits a NLC-decomposition of width \(4\), and,  by Lemma \ref{lem:depth2tree}, one can define the main and auxiliary messages in a way that every node accepts.

For the soundness, let us assume that all nodes of a graph $G$ accept in the verification protocol. From condition~\ref{item:agree-root}, we have that all nodes agree on the same root $r$ of $T_{span}$. From the first item of condition~\ref{item:neighbor-root}, all nodes of depth~1 in $T_{span}$ are adjacent to~$r$. From condition~\ref{item:there-exists-parent}, and from the second item of condition~\ref{item:neighbor-root}, we have that every node at depth \(2\) is adjacent to a node of depth~1. Finally, the third item of condition~\ref{item:neighbor-root} guarantees that the root~$r$ receives all the main messages of the nodes \(G\). Then, by the first item of condition~\ref{item:number-9-number-9}, we have that $r$ can recover all the vertices and all the edges of \(G^*\), and by the second item of condition~\ref{item:number-9-number-9}, we have that \(G^*\) is a cograph.  From conditions \ref{item:pierre-1}-\ref{item:pierre-3}, we get that every edge \(\{u,v\} \in E(G)\) is necessarily included in \(G^*\), meaning that \(G\) is a spanning subgraph of \(G^*\). Finally, the third item of condition~\ref{item:number-9-number-9} guarantees that the number of edges in \(G^*\) equals the number of edges in \(G\), and thus \(G =G^*\). It follows that \(G\) is indeed a cograph.

\subsection{Certificate Size} 

By Lemma~\ref{lem:logdepth}, every cograph admits an NLC-decomposition tree of depth \(\mathcal{O}(\log n)\). Therefore, in the certificate assigned to node~$u$, we have \(d(u) = \mathcal{O}(\log n)\). Each value $x_i(u)$ can be encoded using $O(1)$ bits as  every operation involves a constant number of colors. It follows that  \(\textsf{path}(u)\) can be encoded on $O(\log n)$ bits. The variable \(\textsf{links}(u)\) can be encoded with \( \mathcal{O}(\log n)\) bits as well, by construction. Every node identifier,  and every node degree can be encoded with \( \mathcal{O}(\log n)\) bits. Therefore, all the main messages can be encoded with \( \mathcal{O}(\log n)\) bits. This also implies that the auxiliary messages can be encoded with \( \mathcal{O}(\log n)\) bits. We conclude that, in total, our PLS uses certificates on  \( \mathcal{O}(\log n)\) bits, which completes the proof of Theorem~\ref{theo:cographs}.

%% file: cliquewidth.tex

\section{$\MSO_1$ Properties on Labeled Graphs of Bounded Clique-Width}\label{se:cwd}

This section is dedicated to the proof of Theorem~\ref{theo:main}. To avoid overloading the notation, a labeled graph is simply denoted by~$G=(V,E,\ell)$ where $\ell$ is the vertex-labeling function. We will often simply speak of "graphs" instead of "vertex-labeled graphs". 
Before describing the certificates and the verification protocol, let us first establish some preliminary technical results. 

\subsection{Regularity of $\MSO_1$ predicates}\label{sse:reg}

In their seminal work~\cite{CourcelleMR00}, Courcelle, Makowsky and Rotics proved that any $\MSO_1$ predicate $\Pi$ can be decided in linear time on graphs of bounded clique-width, and hence on graphs of bounded NLC-width, whenever a decomposition tree is part of the input. The running time of the algorithm is linear in the number~$n$ of vertices of the input graph, but the constant hidden in the big-O notation depends on~$k$, on the number of labels and on the $\MSO_1$ formula encoding the predicate $\Pi$. The algorithm in~\cite{CourcelleMR00} is described using tools from automata theory. For our purpose it is more convenient to see it as a dynamic programming algorithm over the decomposition tree of the input graph. Let us formalize this dynamic programming approach following the vocabulary and notations that Borie, Parker and Tovey~\cite{BoPaTo92}. Note that~\cite{BoPaTo92} is alternative proof of Courcelle's theorem on bounded treewidth graphs, specific to a graph grammar defining treewidth; here we simply adapt the definitions to NLC-width and the NLC grammar.

Let $\NLC_k$ denote the class of labeled graphs of graphs of $\NLC$-width at most $k$ and let $\Pi$ be a graph property, assigning to each graph $G$ a boolean value $\Pi(G)$. Intuitively, two graphs $G_1,G_2  \in \NLC_k$ can be considered as equivalent w.r.t. $\Pi$ if, whenever a graph $G'_1 \in \NLC_k$ is obtained by a sequence of NLC operations performed on $G_1$, then the graph $G'_2$ obtained from the same sequence of operations but performed on $G_2$ has the same behaviour w.r.t. property $\Pi$ as $G'_1$. Informally again, a property $\Pi$ is said to be \emph{NLC-regular} if the number of such equivalence classes, that will be called \emph{homomorphism classes} as in~\cite{BoPaTo92}, is upper bounded by a constant: the number of such classes does not depend on the size of the graphs, but only on parameter $k$ and the property $\Pi$ itself. As we shall see, if $\Pi$ is an $\MSO_1$-definable boolean predicate, then it is also NLC-regular, and this is the crux for deciding $\MSO_1$ properties for graphs of bounded NLC-width.

\begin{definition}[NLC-regular property]\label{de:reg}
A graph property $\Pi$ is called \emph{NLC-regular} if, for any value $k$, we can associate a finite set $\cC$ of \emph{homomorphism classes} and a \emph{homomorphism function} $h$, assigning to each graph $G \in \NLC_k$ a class $h(G) \in \cC$ such that:
\begin{enumerate}
\item If $h(G_1) = h(G_2)$ then $\Pi(G_1) = \Pi(G_2)$.
\item For each operation $\Join_S$ there exists a function $\odot_{\Join_S}: \cC \times \cC \rightarrow \cC$ such that, for any two graphs $G_1$ and $G_2$ in $\NLC_k$,
 $$h(\Join_S(G_1,G_2)) = \odot_{\Join_S}(h(G_1),h(G_2)),$$
and for each  operation $\recolor_R$ of there is a function $\odot_{\recolor_R}: \cC \rightarrow \cC$ such that, for any graph $G \in \NLC_k$,
$$h(\recolor_R(G)) = \odot_{\recolor_R}(h(G)).$$
\end{enumerate}
\end{definition}

Observe that NLC-regularity of property $\Pi$ not only implies that the set of homomorphism classes does not depend on the size of the graph, but also that, given an NLC-decomposition of some graph $G \in NLC_k$, the class $h(G)$ can be computed by dynamic programming from the leaves to the root. At each node $x$ of the decomposition tree $T$, the class of $G[x]$, the subgraph of $G$ corresponding to the subtree $T[x]$ rooted at $x$ only depends on the classes at the children nodes, and the operations at node $x$. The first condition of Definition~\ref{de:reg} partitions $\cC$ into a set of accepting classes, i.e., classes $c \in \cC$ such that $h(G) = c \rightarrow \Pi(G)$, and rejecting classes, corresponding to graphs that do not satisfy the property. 

We illustrate Definition~\ref{de:reg} on the predicate $\Pi$ corresponding to non-3-colorability, in order to prove that this predicate is NLC-regular, and how this regularity allows to decide the non-3-colorability on graphs in $\NLC_k$. 
It is  convenient to view a proper 3-coloring of a graph $G$ as a partition of its vertex $V$ set into three independent sets $(X_1, X_2, X_3)$. For such a partition $(X_1, X_2, X_3)$ of a graph $G \in \NLC_k$, $i \in [k]$, we encode this partition as a triple $(b_1,b_2,b_3)$ of boolean vectors of length $k$. Vector $b_j$ encodes the intersection of set $X_j$ with the $k$ possible colours of the NLC decomposition of $G$. That is, $b_j[i] = 1$ if set $X_j$ contains some vertex coloured $i \in [k]$, otherwise $b_j[i] = 0$. Eventually, the homomorphism class $h(G)$ of graph $G$ is the set of all triples $(b_1,b_2,b_3) \in \{0,1\}^k \times  \{0,1\}^k \times \{0,1\}^k$, such that there is some partition of the vertices $G$ into three independent sets $(X_1, X_2, X_3)$ and $(b_1,b_2,b_3)$ is the encoding of this partition as described above. In particular, observe that $G$ is non-3-colourable iff $h(G) = \emptyset$, so $\emptyset$ is the only accepting class.

It is a matter of exercise to see that the number of homomorphism classes is upper bounded by a function on $k$, and even to understand how to construct the class of graphs with a single vertex, and how functions $\odot_{\Join_S}$ and $\odot_{\recolor_R}$ of Definition~\ref{de:reg} can be obtained for all $\Join_S$ and $\recolor_R$ operations. For the sake of completeness, we give the construction in full details.

If $G$ consists of a single vertex $x$ coloured $i$, then the possible partitions of $G$ into three independent sets are $(\{x\},\emptyset,\emptyset)$,  $(\emptyset,\{x\},\emptyset)$ and $(\emptyset,\emptyset,\{x\})$. Thus $h(G)$ is formed by three triples: $(u_{i,k},0^k,0^k)$, $(0^k,u_{i,k},0^k)$, and $(0^k,0^k,u_{i,k})$, where $0^k$ denotes the boolean vector of $k$ zeros, and $u_{i,k}$ is formed of $k-1$ zeros, and a one at position $i$.

Consider now two graphs $G_1,G_2 \in \NLC_k$ and let $G = G_1 \Join_S G_2$ for some $S \in [k] \times [k]$. We describe function $\odot_{\Join_S}$, constructing $h(G)$ from $h(G_1)$ and $h(G_2)$. Note that for each 3-partition $(X_1,X_2,X_3)$ of $G$ into three independent sets, the intersection of $X_1,X_2,X_3$ with the vertex set of $G_1$ (resp. $G_2$) induces a 3-partition $(Y_1,Y_2,Y_3)$ (resp. $(Z_1,Z_2,Z_3)$) into independent sets. Conversely, given partitions  $(Y_1,Y_2,Y_3)$ of $G_1$ and $(Z_1,Z_2,Z_3)$ of $G_2$ into independent sets, $(X_1 = Y_1 \cup Z_1, X_2 = Y_2 \cup Z_2, X_3 = Y_3 \cup Z_3)$ forms a partition into independent sets of $G$ unless operation $\Join_S$ creates an edge between $Y_1$ and $Z_1$, or $Y_2$ and $Z_2$, or $Y_3$ and $Z_3$. Therefore, the homomorphism class of $h(G)$ can be constructed as follows. For each triple of boolean vectors $(y_1,y_2,y_3) \in h(G_1)$ and $(z_1,z_2,z_3) \in h(G_2)$ (corresponding to partitions $(Y_1,Y_2,Y_3)$ of $G_1$ and $(Z_1,Z_2,Z_3)$ of $G_2$ respectively), we add to $h(G)$ the triple $(y_1 \lor z_1, y_2 \lor z_2, y_3 \lor z_3)$ (where $x \lor y$ denotes the bit-wise OR operation), unless there is some pair $(p,q) \in  S$ such that $y_1[p] \land z_1[q]$ or  $y_2[p] \land z_2[q]$ or  $y_3[p] \land z_2[q]$ is true. The latter condition verifies that operation $\Join_S$ does not create an edge in $X_1$ or in $X_2$ or in $X_3$ in graph $G$. 

Eventually, let $G' = \recolor_R(G)$ for some $G \in \NLC_k$, and $R:[k] \to [k]$. We describe function $\odot_{recolor_R}$. Note that $(X_1,X_2,X_3)$ is a partition of $G'$ into independent sets iff it is also a partition of $G$ into independent sets. Therefore we only need to describe how the recoloring $R$ changes the encoding of $(X_1,X_2,X_3)$ from $G$ to $G'$. For each $(b_1,b_2,b_3) \in h(G)$, we add $(\recolor_R(b_1),\recolor_R(b_2), \recolor_R(b_3))$ to $h(G')$, where $b' = \recolor_R(b)$ is defined as follows on boolean vector $b \in \{0,1\}^k$: for each $q \in [k]$ we set $b'[q] = 1$ if and only if there is some $p \in [k]$ such that $R(p) = q$ and $b[p]=1$. In full words, $b'[q]$ is set to true iff the set $X$ encoded by $b$ contains some vertex colored $p$ in $G$, recolored $q$ in $G'$.

This does not only prove that property $\Pi(G)$: ``$G$ is non-3-colourable'' is NLC-regular, but provides all ingredients to decide the property for graphs $G \in \NLC_k$, when a decomposition tree is part of the input. Indeed we need to compute, at each node $x$ of the decomposition tree, the homomorphism class of $G[x]$, the induced subgraph corresponding to thee $T[x]$. At the leaves, graph $G[x]$ is formed by a unique vertex, and as described above the homomorphism class is defined by the colour of that vertex. Then, for each node of the tree, the homomorphism class can be updated from the classes of its children using functions $\odot_{\Join_S}$ and $\odot_{\recolor_R}$. Finally, at the root, we accept if and only if the homomorphism class of the whole graph is the empty set. The running time is linear in the number of nodes of the tree, so it is $O(n)$~--- of course it also depends (exponentially) in parameter $k$.\\


The result of Courcelle, Makowsky and Rotics~\cite{CourcelleMR00}, although expressed in terms of automata theory (see also Theorem 4.2 in~\cite{GanianH10} for an alternative proof), can be restated as follows:

\begin{proposition}\label{pr:reg}
Any graph property $\Pi$ expressible by an $\MSO_1$ predicate is NLC-regular. Moreover, given the corresponding $\MSO_1$ formula and parameter $k$, one can explicitely compute the set of homomorphism classes $\cC$, as well as   functions $\odot_{\Join_S}$ and $\odot_{\recolor_R}$ for any $S \in [k] \times [k]$ and any $R : [k] \to [k]$, and the homomorphism class of the graph formed by a unique vertex coloured $i$, for any $i \in [k]$.
\end{proposition}

Therefore all $\MSO_1$ properties can be decided in $O(n)$ time on graphs of bounded NLC-width, if a decomposition tree is part of the input; again, the big-Oh notation hides a dependency in $k$ and the $\MSO_1$ formula. We also use these ingredients for our proof labeling scheme.

\subsection{General description $\NLC_+$-Width}

Before getting into the full details of our PLS, let us give a general description of the certification algorithm for an $\MSO_1$ property $\Pi$ on graphs of $\NLC$-width at most~$k$. As in the case of cographs (cf. Section~\ref{se:cographs}), given a graph $G \in \NLC_k$, we use an NLC-decomposition tree~$T_{dec}$ of depth $\mathcal{O}(\log n)$ and of width $k' \leq k\cdot 2^{k+1}$, provided by Lemma~\ref{lem:logdepth}. For any vertex $u$ of $G$, recall that $\textsf{path}(u) = (x_1(u), \dots, x_d(u))$ denotes the path in $T_{dec}$ from the root $x_1(u)$ to the leaf $x_d(u)$ corresponding to the node where vertex $u$ is created (again, we abuse notation by identifying the nodes of the tree with the values describing the operations performed in those nodes). 

As in Section~\ref{se:cographs}, the so-called main message of $u$ contains its identifier~$\id(u)$, the sequence $\textsf{path}(u)$, as well as the sequence $\textsf{links}(u)  = (\ell_1(u), \dots, \ell_{d}(u))$, where, for \(i \geq 2\), $\ell_i(u)\in\{0,1\}$ indicates whether $x_{i}(u)$ is the left or right child of $x_{i-1}(u)$ in~$T_{dec}$, when \(x_{i-1}(u)\) is of type \(\Join\) . For certifying a predicate $\Pi$, let $h_i(u)$ denote the homomorphism class at node $x_i(u)$ w.r.t. ${\Pi}$ restricted to graphs of $\NLC$-width at most \(k'\). That is, for $i=1,\dots,d$, 
\[
h_i(u)=h(G[x_i(u)]).
\]
The sequence 
\begin{equation}\label{eq:sequence-h}
\textsf{h}(u) = (h_1(u), \dots, h_d(u))
\end{equation}
is added to the main message of $u$. Note that, since $d \in \mathcal{O}(\log n)$, each sequence is of logarithmic length. Moreover, for every $i\in\{1,\dots,d\}$,  $(x_i(u), \ell_i(u), h_i(u))$ can be encoded on $O(1)$ bits. Therefore the total size of the main message is $\mathcal{O}(\log n)$. 

Unlike the case of cographs, the diameter of a graph of bounded NLC-width is not necessarily bounded by a constant (i.e., Lemma~\ref{lem:depth2tree} does not hold in general for such graphs). It follows that the main messages cannot be gathered in a single vertex as in cographs. To overcome this difficulty, the prover places additional information in  the main message of $u$. First, 
for each color $j \in [k']$ and \(i \in [d(u)]\) let $\textsf{color}_i^j(u)$ denote the number of vertices colored \(j\) in $G[x_i(u)]$ after the recoloring operations performed at node~$x_i(u)$. The sequence 
\[
\textsf{color}^j(u) = (\textsf{color}^j_1(u), \dots, \textsf{color}^j_d(u))
\]
is also added to the main message of $u$, for all $j\in [k']$. 

Finally, we need to guarantee that, for each node \(x\), every vertex that belongs to \(G[x]\) receives the same information about the operations performed in \(x\), as well as the number of vertices colored with each one of the colors.  This verification is especially hairy for nodes \(x\) where \(G[x]\) is disconnected, where this consistency verification is done with the help of vertices outside \(G[x]\).  In order to cope with this issue, let us slightly modify the notion of NLC-decomposition tree by allowing nodes~$\parallel$ (i.e., nodes of $T_{dec}$ at which a disjoint union operation is performed) to be of arbitrarily large arity (in the original decomposition $T_{dec}$, the arity of an internal node is~2). Moreover all disconnected graphs $G[x]$ will correspond to parallel nodes $x$, and for the parent $y$ of such a node, $G[y]$ has to be connected. Such a decomposition tree will be called an $\NLC_+$-decomposition tree. 


\begin{definition}\label{def:nlc+}
Let $k\geq 1$. A rooted tree $T$ is an $\NLC_+$-decomposition tree of width \(k\) if the following conditions holds:  
 \begin{enumerate}
 \setlength\itemsep{0em}
 \item Every leaf of $T$ is labeled $\mathsf{newVertex}_i$, for some \(i \in [k]\);
 \item Every internal node of \(T\) is labelled \(\Join\) or \(\parallel\);
\item Every node labeled  \(\Join\) has exactly \(2\) children, and such a node is associated to a set $S\in [k]\times[k]$ and to a function $\mathsf{recolor}_R$ where $R:[k] \to[k]$;
\item Every node labeled \(\parallel\) has at least 2 children, and, for every node labeled \(\parallel\) distinct from the root, its parent is labeled \(\Join\);
\item \label{item:connectivityNLC+} Every graph defined by the subtree rooted in a \(\Join\) node is connected.
\end{enumerate}
\end{definition}

In Definition~\ref{def:nlc+}, the nodes labeled  \(\Join\)  represent the join and recoloring operations performed as in \(\NLC\)-decomposition trees, i.e., first join, and then recolor.  Similarly, the nodes \(\parallel\)  represent the disjoint union of the graphs defined by their children. However, a crucial difference compare to $\NLC$-decomposition is that instead of systematically involving two vertex-disjoint graphs, the $\NLC_+$-decomposition allows an arbitrary number of vertex-disjoint graphs. Another crucial difference with $\NLC$-decomposition is Condition~\ref{item:connectivityNLC+}, which imposes connectivity of the subgraph hanging at every  \(\Join\) node. 

To cope with the notion of $\NLC_+$-decomposition, we merely extend the functions $\odot$ given in Definition~\ref{de:reg} by introducing the operator \(\odot_{\parallel}\) with arbitrary arity. Specifically, let \(G\) be a graph obtained from the disjoint union of a set of \(\NLC_+\) graphs $\{G_1, \dots, G_p\}$, in any order, with \(p\geq 2\). We define 
\begin{align*}
h(G)  =  \odot_{\parallel }(h(G_1), \dots, h(G_p)) = \odot_{\parallel }\Bigg(h(G_1), \odot_{\parallel }\bigg(h(G_2), \odot_{\parallel }\Big(\dots, \odot_{\parallel }\big(h(G_{p-1}), h(G_p)\big)\Big)\bigg)\Bigg). 
\end{align*}
We say that a graph has \(\NLC_+\) width \(k\) if it can be constructed according to an \(\NLC_+\) decomposition tree of width \(k\). For any node \(x\) of an \(\NLC_+\) decomposition tree \(T\), we also define \(T[x]\) and \(G[x]\) in the same way as for NLC-decompositions trees. We now show that allowing large arity and imposing connectivity does not ruin the good properties of $\NLC$-decomposition.

\begin{lemma}\label{lem:ldc} 
For every $k\geq 1$, all \(n\)-node connected graphs of $\NLC$-width \(k\) have \(\NLC_+\)-decomposition trees of width at most \(k\cdot 2^{k+1}\), and depth \(\mathcal{O}(\log n)\).
\end{lemma}

The proof of the  lemma is based on the following statement.

\begin{claim}\label{claim:pourlem:ldc} 
Let $G=(V,E)$ be a connected graph of $\NLC$-width $w\geq 1$. Let $T$ be an $\NLC$-decomposition tree of $G$ of width~$w$ and depth~$d$. Then $G$ admits an \(\NLC_+\)-decomposition tree $T_+$  of width~$w$ and depth at most~$2d$. Moreover, the color of each vertex of $G$ at the root of $T_+$ is the same as the color of this vertex at the root of~$T$.
\end{claim}

\begin{proof}
The proof is by induction on $d$. 
The claim is straightforward for $d=1$ as  $T$ is also an \(\NLC_+\)-decomposition tree in this case. Let $d>1$, and let us assume by induction hypothesis that the claim holds for all connected graphs having an NLC-decomposition tree of depth smaller than $d$. Let $x$ be a node of tree $T$ such that $G[x]$ (corresponding to the decomposition subtree $T[x]$) is connected, but, for at least one of the children $z_1,z_2$ of $x$, $G[z_i]$ is disconnected. Note that if no such $x$ exists, then $T$ is an \(\NLC_+\)-decomposition tree, the connectivity condition being satisfied at each node. We choose $x$ closest to the root, in the sense that no other node from $x$ to the root has this property. (Of course, there might be several such nodes $x$, none of them being ancestor of the other in the tree.)
For each $i \in \{1,2\}$, let $D_i^1, D_i^2,\dots, D_i^{p_i}$ be the connected components of $G[z_i]$. Observe that each $D_i^j$, $1\leq j \leq p_i$ has an $\NLC$-decomposition tree $T_i^j$ of depth at most the depth of $T[z_i]$, obtained by trimming from $T[z_i]$ all leaves that do not correspond to vertices of $D_i^j$. By induction hypothesis, graph $D_i^j$ has an \(\NLC_+\)-decomposition tree $T_+(i,j)$ of width $w$ and of depth at most twice the depth of $T[z_i]$.

We obtain the tree $T_+$ as follows. For node $x$, if $G[z_i]$ has a unique component $D_i^1$, we replace the subtree $T[z_1]$ by $T_+(i,1)$. If $G[z_i]$ has several components $D_i^1, D_i^2,\dots, D_i^{p_i}$, then we replace $T[z_i]$ by a subtree of root $\parallel$, with $p_i$ children, the $j$th child being the root of $T_+(i,j)$. Observe that in this way, the label of $x$ remains unchanged, and $T_+[x]$ is an  \(\NLC_+\)-decomposition tree of $G[x]$. By performing in parallel these operations on all such nodes $x$, we obtain an  \(\NLC_+\)-decomposition tree of $G$. The width of the decomposition has not changed. Moreover the depth has been increased by a factor~2 at most. Indeed, for each node $x$, we replaced $T[z_i]$ by subtrees of depth at most twice the original depth, and the node $x$ itself might have caused the addition of a new layer of children labelled $\parallel$. Therefore, the depth of $T_+[x]$ is at most $2 \cdot \max\{\mbox{depth}(T[z_1]),\mbox{depth}(T[z_2])\}$+2, hence at most twice the depth of $T[x]$. Altogether, the depth of $T_+$ is at most twice the depth of $T$. This concludes the induction step, and the proof of  Claim~\ref{claim:pourlem:ldc}. 
\end{proof}

\begin{proof}[Proof of Lemma~\ref{lem:ldc}.]
Thanks to Lemma~\ref{lem:logdepth}, every $n$-node connected graph $G$ of NLC-width $k$ admits an NLC-decomposition tree of width $k \cdot 2^{k+1}$ and depth $O(\log n)$. By Claim~\ref{claim:pourlem:ldc}, we obtain an \(\NLC_+\)-decomposition tree of the same width $k \cdot 2^{k+1}$, and still of logarithmic depth. 
\end{proof}

Let \(T\) be an $\NLC_+$ -decomposition. Recall that, for every node \(x\) of \(T\), we denote by \(h(x)\) the homomorphism class of \(G[x]\) w.r.t. property \({\Pi}\). Moreover, we set 
\begin{equation}\label{eq:sequence-color}
\textsf{color}(x)=(\textsf{color}^1(x), \dots, \textsf{color}^{k'}(x))
\end{equation}
where, for every $j\in \{1, \dots, k'\}$, $\textsf{color}^j(x)$ is the number of vertices colored \(j\) at node~$x$ after the recoloring operations. Finally, we denote by \(\textsf{exit}(x)\) the identifier of some vertex that belongs to \(G[x]\), called the \textit{exit vertex} of \(G[x]\). When \(x\) is of type \(\parallel\) then \(\textsf{exit}(x)\) is an arbitrary vertex in \(G[x]\).  When \(x\) is of type \(\Join\), we have by Definition~\ref{def:nlc+} that \(G[x]\) is connected. Let \(z^0\) and \(z^1\) be the left and right children of \(x\), respectively. Then, we choose \(\textsf{exit}(x)\) as an arbitrary node belonging to \(G[z^0]\) that is adjacent to some node in \(G[z^1]\).


\subsection{Certificate Assignment }  

As for cographs, the certificates are divided in several parts, called \emph{main messages} and \emph{auxiliary messages}. We add a third part, called \emph{service messages}. Let us fix some \(\NLC_+\)-decomposition of width \(k\) and depth \(\mathcal{O}(\log n)\) of the input graph \(G\). 

\subparagraph{Main messages.}  
These messages are used to check the local correctness of the decomposition tree. The main message of node \(u \in V\) contains the following information. 

\begin{itemize}

\item The sequence \(\textsf{path}(u) = (x_1(u), \dots, x_{d}(u))\) as defined in Section~\ref{se:cographs} for certifying cographs.

\item The sequence \(\textsf{links}(u) = (\ell_1(u), \dots, \ell_{d}(u)) \in \{0,1,\bot\}^d\), representing the sequence of edges that are to be followed to reach \(x_d(u)\) from \(x_1(u)\), similarly to Section~\ref{se:cographs}, but taking into account the presence of nodes  \(\parallel\) with large arity. More precisely, \(\ell_1(u) = \bot\), and for every $i\geq 2$,
\begin{itemize}
\item if \(x_{i-1}(u)\) is of type \(\parallel\) then \(\ell_i(u) = \bot\);
\item if \(x_{i-1}(u)\) is of type \(\Join\) then \(\ell_i(u)= 0 \) whenever \(x_{i}(u)\) is the left children of \(x_{i}(u)\), and  \(\ell_i(u)= 1\) otherwise.  
\end{itemize}

\item The sequence $\textsf{h}(u) = (h_1(u), \dots, h_d(u))$ defined in Eq.~\eqref{eq:sequence-h}.

\item The sequence $\textsf{color}(u) = (\textsf{color}_1(u), \dots, \textsf{color}_d(u))$, such that, for each \(i \in \{1, \dots, d\}\), \(\textsf{color}_i(u)\) is the sequence $\textsf{color}(x_i(u))=(\textsf{color}^1_i(u), \dots,\textsf{color}^{k'}_i(u))$ defined in Eq.~\eqref{eq:sequence-color}.

\item The sequence \(\textsf{exit}(u) = (\textsf{exit}_1(u), \dots, \textsf{exit}_d(u))\), where, for each \(i \in \{1, \dots, d\}\), \(\textsf{exit}_i(u)\) is the identifier of node \(\textsf{exit}(x_i(u))\). 
\end{itemize}

\noindent In the following, for every \(i \in [d]\), we denote
 $
 \textsf{main}(u)= (\textsf{main}_1(u), \dots, \textsf{main}_d(u))
 $
 where, for each \(i \in [d]\),
\[ 
\textsf{main}_i(u)  = \Big( x_i(u),\; \ell_i(u),\; h_i(u),\; \textsf{color}_i(u),\; \textsf{exit}_i(u) \Big).
\]

\subparagraph{Auxiliary messages.} 
These messages are used to certify the connectivity of the subtrees rooted at nodes \(x\) of type \(\Join\). Let \(z^0\) and \(z^1\) be the children of \(x\) in \(T\). An \emph{auxiliary tree associated to \(x\)}, denoted by~\(A(x)\),  is a spanning tree of \(G[x]\). In the auxiliary messages, we certify the existence of \(A(x)\) using the standard certification for trees~\cite{KormanKP10}. That is, we give each node of the tree the identifier of the root, the identifier of its parent, and its distance to the root.  We also use the auxiliary messages to certify that each vertex created in \(G[z^0]\) (respectively \(G[z^1]\)) received the same information about \(z^0\). Formally, every vertex~\(u\) receives the sequence 
\[
\textsf{aux}(u) = (\textsf{aux}_1(u), \dots, \textsf{aux}_{d}(u))
\]
where, for each \(i \in \{1, \dots, d\}\), \(\textsf{aux}_i(u)=\bot\) whenever \(x_i(u)\) is of type $\parallel$, and, otherwise, 
\[
\textsf{aux}_i(u)= (\textsf{root}_i(u), \textsf{parent}_i(u), \textsf{distance}_i(u), \textsf{childrenMain}_i(u))
\] 
where 
\begin{itemize}
 \setlength\itemsep{0em}
\item \( \textsf{root}_i(u)\) is the identifier of the node \(\textsf{exit}(x_i(u))\);
\item  \( \textsf{parent}_i(u)\) is the identifier of the parent of \(u\)  in \(A(x_i(u))\), where  \(\textsf{parent}_i = \bot\) if $u = \textsf{root}_i(u)$;
\item \( \textsf{distance}_i(u)\) is the distance  between \(u\) and \(\textsf{exit}(x_i(u))\)  in \(A(x_i(u))\);
\item \(\textsf{childrenMain}^j_i(u)=(z^j, h(z^j), \textsf{color}(z^j), \textsf{exit}(z^j))\) for \(j \in \{0,1\}\). 
\end{itemize}

\subparagraph{Service messages.}  These messages are used to check consistency in the subgraphs induced by the subtrees rooted at nodes of type \(\parallel\). In other words, they are used to handle the case of nodes in the tree~$T$ constructing non-connected subgraphs. Before explaining these messages, let us define some additional data structures. 

Let \(x\) be a node of type \(\parallel\), and let $y$ denote the parent of $x$ in $T$. Again by Lemma~\ref{lem:ldc}, for each child $z_i$, $1 \leq i \leq p$, of $x$,  the graph $G[z_i]$ corresponds to a connected component of the graph $G[x]$ (which is disconnected by construction). A \emph{service tree associated to} \(x\), denoted by \(S(x)\), 
is a Steiner tree in \(G[y]\) with terminals \(\{\textsf{exit}(z_1), \dots, \textsf{exit}(z_p)\}\), and root \(\textsf{exit}(x)\). The service messages are used to certify the existence of \(S(x)\) in a  way similar to the auxiliary tree.  We also use the service messages  to certify the consistency between \(h(z_1), \dots, h(z_p)\) and \(h(x)\).  This latter certification is slightly more complicated than for auxiliary trees because a parallel node may have an arbitrarily large number of children, and one cannot store  the identifiers, the classes and the colors of all the terminals if one want to keep the certificate size small.

Let \(u\) be a vertex in \(S(x)\). The service message of \(u\) contains the root of \(S(x)\), the parent of \(u\), and the depth of \(u\) in \(S(x)\). Let us denote by \(S(x,u)\) the subtree of \(S(x)\) rooted at~\(u\).  Furthermore, let us call \(\textsf{inCharge}_x(u)\) the set of the indices of the terminals of \(S(x)\) contained in \(S(x,u)\). Formally,
\[
\textsf{inCharge}_x(u) = \{i \in \{1, \dots, p\} \mid \textsf{exit}(z_i) \in S(x,u) \} 
\]
Now let us denote by \( G[ \textsf{inCharge}_x(u) ] \) the subgraph of \(G\) induced by the disjoint union of all graphs \(G[z_i]\), \(i \in \textsf{inCharge}_x(u)\). The service message of \(u\) includes the homomorphism class \(h(G[ \textsf{inCharge}_x(u) ])\). The correctness of this homomorphism class will be verified using only  the homomorphism classes of the children of \(u\) in \(S(x)\).  Observe that \(S(x, \textsf{exit}(x)) = S(x)\), and \[h(G[ \textsf{inCharge}_x(\textsf{exit}(x))] = h(x).\]

Observe that node  \(u\) is necessarily contained in \(G[y]\), but not necessarily in \(G[x]\). Therefore, it is possible that \(x\) does not appear in the main message of \(u\). Therefore, for each \(j \in \{0,1\}\), vertex \(u\) also receives the sequence
\[
\textsf{service}^j(u) = (\textsf{service}^j_1(u), \dots, \textsf{service}^j_{d}(u)), 
\]
where, for each \(i \in [d]\), the value of \(\textsf{service}^j_i(u)\) represents the part of the certification of the service tree \(S(x)\), where \(x\) is the left child (respectively the right child) of \(x_{i-1}(u)\). 
Specifically, if \(x\) is not of type \(\parallel\), or if \(u\) does not participate in \(S(x)\), then \(\textsf{service}^j_i(u)=\bot\). Otherwise, 
\[
\textsf{service}^j_i(u) = (\textsf{root}^j_i(u), \textsf{parent}^j_i(u), \textsf{distance}^j_i(u), \textsf{class}^j_i(u), \textsf{colorCharge}^j(u))
\]
where
\begin{itemize}
 \setlength\itemsep{0em}
\item \( \textsf{root}^j_i(u)\) is the identifier of the root of \(S(x)\); 
\item  \( \textsf{parent}^j_i(u)\) is the identifier of the parent of \(u\)  in \(S(x)\), with \(\textsf{parent}^j_i (u)= \bot\) if \(u = \textsf{root}^j_i(u)\),;
\item \( \textsf{distance}^j_i(u)\) is the distance between \(u\) and the root \(\textsf{exit}(x)\)  in \(S(x)\);
\item \(\textsf{class}^j_i(u)\) is the homomorphism class of \(G[\textsf{inCharge}_{x}(u)]\);
\item  $\textsf{colorCharge}_i^j(u)=(\textsf{colorCharge}_i^{j,1}(u), \dots, \textsf{colorCharge}_i^{j,k'}(u))$ where $\textsf{colorCharge}_i^{j,s}(u)$ is the number of vertices colored $s$ in $G[\textsf{inCharge}_{x}(u)]$, for every $s \in [k']$.
\end{itemize}

\subsection{Verification Scheme} 
The main role of the verification procedure is to check that  the \(\NLC_+\)-decomposition tree is correct, that is (1)~it corresponds to the graph $G$, and (2)~at each node $x$ of the decomposition tree, the homomorphism class for property $\Pi$ provided by the prover corresponds to $h(G[x])$.

Each vertex \(u \in V\) checks that all its messages in its certificate are correctly formatted.  Then, the verification is split in three parts: (1)~verification of the main messages, (2)~verification of the auxiliary messages, and (3)~verification of the service messages. 
 
\subparagraph{Verification of main messages. }

Each vertex \(u \in V\) verifies the following conditions:

\begin{enumerate}

\item \label{verif:1} \(x_{d}(u)\) is of type $\mathsf{newVertex}_s$ for some \(s \in [k']\);
\item  \label{verif:2} $\textsf{exit}_d(u)=\id(u)$;
\item  \label{verif:3} \(\textsf{color}^j_d(u) = 0\) for every \(j\in [k']\setminus\{s\}\), and \(\textsf{color}^s_d(u) = 1\);
 \item  \label{verif:4} \(h_d(u)\) is the homomorphism class for \({\Pi}\) of a graph equal to a single node labelled $\ell(u)$, colored $s$;
 \item  \label{verif:5} \(x_1(u)\) is of type~\(\Join\), and for each \(i\in \{2, \dots, d(u)\}\), if \(x_i\) is of type \(\parallel\) then \(x_{i-1}(u)\) is of type \(\Join\);
 \item \label{verif:5.1} \(h_1(u)\) is an \emph{accepting} homomorphism class for \({\Pi}\), i.e., graphs of \(\NLC_{+}\)-width \(k'\) having this class satisfy~\(\Pi\).
\end{enumerate}
Let us define \(\mathsf{currentcolor}(u,i)\) in the same way that we did in the verification of cographs. Moreover, for a vertex \(v\in V\), let us denote \(i^*=\textsf{index}(u,v)\) as the maximum index \(i\in [d(u)]\) such that \(\ell_j(u) = \ell_j(v)\) for every \(j \leq i\). Then, \(u\) computes \(\textsf{index}(u,v)\) for each of its neighbors \(v\in N(v)\), and checks:

\begin{enumerate}
\setcounter{enumi}{6}
\item \label{verif:6} \(\textsf{main}_i(u) = \textsf{main}_i(v)\) for every \(i \leq i^*\).

\item \label{verif:7} \(x_{i^*}(u)\) is of type \(\Join\), and 
\begin{itemize}
\item   if \(\ell_{i^*+1}(u) = 0\) then \(u\) checks that \(x_{i^*}\) contains the join operation $\Join_S$ with \[(\mathsf{currentcolor}(u,i^*), \mathsf{currentcolor}(v,i^*)) \in S\]
\item   if \(\ell_{i^*+1}(u) = 1\) then \(u\) checks that  \(x_{i^*}\) contains the join operation $\Join_S$ with \[(\mathsf{currentcolor}(v,i^*),\mathsf{currentcolor}(u,i^*)) \in S.\]
\end{itemize}

\item \label{verif:11} For each \(j \in [k']\), vertex \(u\) checks that it has exactly \(\textsf{color}^j_{i^*+1}(v)\) neighbors \(w\) such that \(\textsf{index}(u,w) = i^*\), \(\ell_{i^*+1}(w) \neq \ell_{i^*+1}(u)\), and \(\mathsf{currentcolor}(w,i^*+1) = j\);


\item \label{verif:9} Vertex $u$ then considers the homomorphism classes $h_{i^*}(u)$, $h_{i^*+1}(u)$ and $h_{i^*+1}(v)$. (Recall that these encode the homomorphism classes of relation ${\Pi}$ for graphs $G[x_{i^*}(u)]$, $G[x_{i^*+1}(u)]$ and $G[x_{i^*+1}(v)]$, respectively.) 
\begin{itemize}
\item If \(\ell_{i*+1} = 0\) then $x_{i^*+1}(u)$ must be the left child of $x_{i^*}(u)$, and vertex  $u$ checks that \[h_{i^*}(u) = \odot_{\mathsf{recolor}_R} (\odot_{\Join_S}(h_{i^*+1}(u), h_{i^*+1}(v))).\] 
\item If  \(\ell_{i*+1} = 1\) then $x_{i^*+1}(u)$ must be the right child of $x_{i^*}(u)$, and vertex  $u$ checks  that \[h_{i^*}(u) = \odot_{\mathsf{recolor}_R} (\odot_{\Join_S}(h_{i^*+1}(v)),h_{i^*+1}(u)).\] 
\end{itemize}

\item \label{verif:10}  For each \(j\in [k']\), \[\textsf{color}^j_{i^*}(u) = \sum_{s \in R^{-1}(j)}  \left(\textsf{color}^{s}_{i^*+1}(u) + \textsf{color}^{s}_{i^*+1}(v)\right).\]

\item  \label{verif:12} Finally, for each \(i \in [d(u)]\) such that \(x_{i}(u)\) is of type \(\Join\), if \(\textsf{exit}_i(u)=\id(u)\), then \(u\) checks that \(\ell_{i+1}(u) = 0\), and that there exists \(v\in N(u)\) such that \(\textsf{index}(u,v) = i\) and \(\ell_{i+1}(v) = 1\).
\end{enumerate}

\subparagraph{Verification of auxiliary messages.}  

For every \(i \in [d(u)]\) vertex \(u\) checks that \(\textsf{aux}_i \neq \bot \) if and only if \(x_i(u)\) is of type \(\Join\). Let us suppose now that \(\textsf{aux}_i \neq \bot\). Then \(u\) checks the following conditions

\begin{enumerate}
\setcounter{enumi}{12}
\item \label{verif:13} \(\textsf{root}_i(u) = \textsf{exit}_i(u)\);
\item  \label{verif:14} if \( \textsf{parent}_i(u) = \perp\) then \(\textsf{root}_i(u)=\id(u)\);
 \item  \label{verif:15} if \(v = \textsf{parent}_i(u) \neq \bot\) then \(v \in N(u)\), \(\textsf{index}(u,v) \geq i\), \(\textsf{aux}_i(v) \neq \perp\), \(\textsf{root}_i(v) = \textsf{root}_i(u)\), \(\textsf{childrenMain}^0_i(u) = \textsf{childrenMain}^0_i(v) \),  \(\textsf{childrenMain}^1_i(u) = \textsf{childrenMain}^1_i(v) \), and \[\textsf{distance}_i(v) = \textsf{distance}_i(u) -1;\]
\item  \label{verif:16} if \(j = \ell_{i+1}(u) \) then \(\textsf{childrenMain}^j_i(u) = \textsf{main}_{i+1}(u)\).

\end{enumerate}

\subparagraph{Verification of service messages.}   

For every \(i \in [d(u)]\) and \(j \in \{0,1\}\) for which \(\textsf{service}^j_i(u)\neq \perp\), \(u\) checks the following conditions:

\begin{enumerate}
\setcounter{enumi}{16}
\item \label{verif:17}  If \(\ell_i(u) = j\) then \(x_i(u)\) is of type \(\parallel\).

\item  \label{verif:18} If \(v = \textsf{parent}^j_i(u) \neq \perp\) then \(v \in N(u)\), \(\textsf{index}(u,v) \geq i-1\), \(\textsf{service}^j_i(v) \neq \perp\), \(\textsf{root}^j_i(v) = \textsf{root}^j_i(u)\) and \(\textsf{distance}^j_i(v) = \textsf{distance}_i^j(u) -1\).
\item  \label{verif:19} If \( \textsf{parent}^j_i(u) = \perp\) then \(\textsf{root}^j_i(u)=\id(u)\).

\end{enumerate}
If \(u\) has no neighbors \(w \in V\) such that \(\textsf{parent}^j_i(w) = u\), then \(u\) deduces that it is a leaf of the tree \(S(x_i^j)\). In that case, \(u\) checks the following conditions:
\begin{enumerate}
\setcounter{enumi}{19}

\item \label{verif:20}  \(u\) is a terminal, so \(\textsf{exit}_{i+1}(u)=\id(u)\);
\item \label{verif:21} \(\textsf{class}_i^j(u) = h_{i+1}(u)\) and \(\textsf{colorCharge}_i^j(u) = \textsf{color}_{i+1}(u)\).

\end{enumerate}
If $u$ is not a leaf of \(S(x_i^j)\), then \(u\) computes the set \(\textsf{children}_i^j(u)\) of nodes  \(w\in V\) such that \({\textsf{parent}_i^j(w)=u}\). Then \(u\) checks the following conditions:

\begin{enumerate}
\setcounter{enumi}{21}
\item \label{verif:23} For each \(s\in [k']\), let \(\gamma = \textsf{color}_{i+1}^s(u)\) if \(u = \textsf{exit}_{i+1}(u)\) and \(\ell_i(u) = j\), and let  \(\gamma = 0\) otherwise.  Then \(u\) checks that: 

\[\textsf{colorCharge}^{j,s}(u)  = \sum_{w \in \textsf{children}_i^j(u)} \textsf{colorCharge}^{j,s}(w) + \gamma\]

\item \label{verif:24} If \(\textsf{children}_i^j(u)\) contains a single vertex \(w\), then \(u\) checks that \(\textsf{class}_i^j(u) = \textsf{class}_i^j(w)\).  
\item \label{verif:25} If \(\textsf{children}_i^j(u)\) contains two or more vertices, then \(u\) defines an arbitrary order of the vertices in \(\textsf{children}_i^j(u)\), namely \(w_1, \dots, w_p\). Then, \(u\) defines a sequence of homomorphism classes \(c_1, \dots, c_{p}\) where, \(c_1 = \textsf{class}_i^j(w_1)\), and for each \(i \in \{2, \dots, p\}\):
\[c_i = \odot_{\parallel} (c_{i-1},  \textsf{class}_i^j(w_{i})).\]
where \(\odot_{\parallel}\) corresponds to the function \(\odot_{\Join_{\emptyset}}\). Finally checks that \(c_p = \textsf{class}_i^j(u) \).
\end{enumerate}
 If \(u = \textsf{exit}_{i}(u)\), then \(u\) checks the following additional conditions for each \(j \in \{0,1\}\):
\begin{enumerate}
\setcounter{enumi}{24}
\item \label{verif:26} \(\textsf{parent}_i^j(u) = \bot\), and  \(\textsf{distance}_i^j(u) = 0\);
\item \label{verif:27} \(\textsf{class}_i^j(u) = h_{i}(u)\), and \(\textsf{colorCharge}_i^j(u) = \textsf{color}_{i}(u)\).
\end{enumerate}


%
\subsection{Completeness and Soundness} 
The completeness follows directly from the existence of an $\NLC_+$-decomposition $T$ of $G$ as described in Lemma~\ref{lem:ldc} and from the construction of the main, auxiliary and service messages, based on this decomposition. 

For soundness, assume that the verification protocol accepts at every vertex.  Let us define the set of all main messages as \(\textsf{mainset} = \{\textsf{main}(u) : u \in V\}\), and let \(d^* = \max_{u\in V} d(u)\). For every \(M \in \textsf{mainset}\) we denote \(d(M) = d(u)\) where \(u\) is such that \(\textsf{main}(u) = M\). Now, we define  \(\textsf{Prefix}(M,i)\) as the set of vertices \(v\in V\) that have the same first \(i\) values in their main messages. Formally,  
\[
\textsf{Prefix}(M,i) = \{v \in V: M_j = \textsf{main}_j(v),  \textrm{ for every } 1\leq  j \leq  i\}.
\] 
By condition~\ref{verif:2} in the verification protool, \(\textsf{exit}_{d}(u)=\id(u)\). Thus \(\textsf{main}(u) \neq \textsf{main}(v)\) for every pair of vertices \(u \neq v\). Therefore, \(\textsf{Prefix}(\textsf{main}(u),d(u)) = \{u\}\). Also, by condition~\ref{verif:6} we have that \(\textsf{Prefix}(M,1) = V\). Given \(M \in \textsf{mainset}\),  and \(i \in [d(M)]\), let 
\[
M_i = (x_i(M), \ell_i(M), h_i(M), \textsf{color}_i(M), \textsf{exit}_i(M))
\] 
be the list \(\textsf{main}_i(v)\) for a vertex \(v\) such that \(\textsf{main}(v) = M\).

\begin{lemma}\label{lem:completeness} 
For every \(M \in \textsf{mainset}\), and every \(i \in [d(M)]\), there exists a \(\NLC_+\)-decomposition tree \(T[M,i]\) of \(G[M,i] = G[\textsf{Prefix}(M,i)]\)  such that, for all \(u \in \textsf{Prefix}(M,i)\), the following holds. 
\begin{itemize}
\item For every $j$ with \(i \leq j \leq d(u)\), \(x_j(u)\) contains the operations in the \(j\)-th node in path from the root of \(T[M,i]\) to the node where \(u\) is created;
\item \(h_i(u) = h(G[M,i])\);
\item  For every \(s\in [k']\), \(\textsf{color}^s_i(u) \) is the number of vertices colored \(s\) in the root of \(T[M,i])\).
\end{itemize}
\end{lemma}

\begin{proof}
The proof is by induction on \(i\), in decreasing order for each \(M\). Let us fix \(M \in \textsf{mainset}\). The base case is \(i = d(M)\), and \(\textsf{Prefix}(M,d(M))\) is a single vertex \(u\) satisfying \(\textsf{main}(u) = M\).  In this case the lemma is true by conditions~\ref{verif:1}-\ref{verif:4}. For the inductive case, let us suppose that there exists \(t >1 \)  such that the lemma is true for every \(M \in \textsf{mainset}\), and \(i \in \{t, \dots, d(M)\}\), and let us show that the lemma holds an arbitrary  pair \(M \in \textsf{mainset}\) and \(i = t-1\). 

\medskip

\noindent -- Let us suppose first that \(x_{i}(M)\) is of type \(\Join\). Consider the set
\[
\beta_i(M) =  \{(\textsf{root}_i(u), \textsf{parent}_i(u),  \textsf{distance}_i(u)) \mid u \in \textsf{Prefix}(M,i)\}.
\] 
From conditions~\ref{verif:13}-\ref{verif:15}, we have that \(\beta_i(M)\) certifies a spanning tree of \(\textsf{Prefix}(M,i)\) rooted at \(\textsf{exit}_i(M)\). In particular, we have that \(\textsf{Prefix}(M,i)\) induces a connected subgraph of \(G\). Let us call \(u^0\) the vertex with identifier \(\textsf{exit}_i(M)\). By condition~\ref{verif:12}, we have \(\ell_{i+1}(u^0) = 0\), and there exists a vertex \(u^1 \in N(u)\) such that \(\textsf{index}(u,v) = i\) and \(\ell_{i+1}(u^1) = 1\). Let us call \(M^0= \textsf{main}(u^0)\) and \(M^1 = \textsf{main}(v^1)\). Our candidate for \(T[M,i]\) is the tree defined by a root \(x_i(M)\) with two children. The left children induces the subtree \(T[M^0, i+1]\) while the right children induces the subtree \(T[M^1,{i+1}]\).

Observe first that \(\textsf{Prefix}(M, i)\) can be partitioned in \(\textsf{Prefix}(M^0, i+1)\) and  \(\textsf{Prefix}(M^1, i+1)\). Indeed, let \(w\) be a vertex in \(\textsf{Prefix}(M, i)\), and let \(j =\ell_{i+1}(w) \). By condition~\ref{verif:16},   we have that \(\textsf{childrenMain}^j_i(w)=\textsf{main}_{i+1}(w)\). By condition~\ref{verif:15}, and the fact that \(\textsf{Prefix}(M, i)\) is connected, we have that \(\textsf{main}_{i+1}(w) = \textsf{main}_{i+1}(w')\) for every \(w' \in \textsf{Prefix}(M, i)\) for which \(\ell_{i+1}(w') = j\). In particular \(\textsf{main}_{i+1}(w) = \textsf{main}_{i+1}(u^j)\), and therefore \(w\)  belongs to \(\textsf{Prefix}(M^j, i+1)\).

By the induction hypothesis, for each \(j\in \{0,1\}\), \(T[M^j, i+1]\) is an \(\NLC_+\) decomposition tree of \(G[M^j, i+1]\). Moreover, \(h_{i+1}(M^j) = h(G[M^j, i+1])\), and, for each \(s\in [k']\), we have that \(\textsf{color}^s_{i+1}(M^j)\) is the number of vertices colored \(s\) in the root of \(T[M^j, i+1]\).
Let \(S\) be the set of join operations defined in \(x_i(M)\), and let \(u \in \textsf{Prefix}(M^0,i+1)\) and \(v \in \textsf{Prefix}(M^1,i+1)\) be such that \(\{u,v\} \in E(G[\textsf{Prefix}(M,i)])\). By condition~\ref{verif:7}, we have that 
\[
(\textsf{currentcolor}(u,i), \textsf{currentcolor}(v,i)) \in S.
\] 
Moreover, by condition~\ref{verif:11}, vertex \(u\) has exactly \(\textsf{color}_{i+1}^s(M^1)\) neighbors  \(v \in \textsf{Prefix}(M^1,i+1)\) such that \(\textsf{currentcolor}(v,i) = s\). This implies that  the join operations defined in \(x_i(M)\) create exactly the set of edges \(\{u,v\} \in E(G[M,i])\) such that \(\textsf{index}(u,v) = i\), and no other edges. It follows that \(T[M,i]\) is an \(\NLC_+\)-decomposition tree of \(G[M,i]\). Finally, by condition~\ref{verif:9} applied at \(u^0\), we have that \(h_i(M) = h(G[M,i])\), and, by condition~\ref{verif:10} applied at \(u^0\), we have that, for every \(s\in [k']\), \(\textsf{color}_i^s(M)\) is the number of vertices colored \(s\) in the root of \(T[M,i]\).

\medskip 

\noindent -- Let us suppose now  that \(x_{i}(M)\) is of type \(\parallel\), and let us define the following subset of \textsf{mainset}:
\[\textsf{mainPrefix}(M,i) = \{ M' \in \textsf{mainset} \mid \textsf{Prefix}(M,i) = \textsf{Prefix}(M',i) \}\]
Observe that from the induction hypothesis, we have that, for each \(M' \in \textsf{mainPrefix}(M,i)\), there is an \(\NLC_{+}\)-decomposition tree  \(T[M',i+1]\) of \(G[M',i+1]\), satisfying the conditions of the lemma.
Our candidate for \(T[M,i]\) is a tree rooted in  \(x_{i}(M)\), where the children of \(x_{i}(M)\) induce the set of trees 
\[
\{ T[M', i+1]: M' \in \textsf{mainPrefix}(M,i) \}.
\]
Indeed,  every vertex \(u\) in \(\textsf{Prefix}(M,i)\) satisfies that \(x_i(u) = x_i(M)\). Then, by condition~\ref{verif:7},  \(u\) cannot have a neighbor \(v\) such that \(\textsf{index}(u,v) = i\). Therefore \(T[M,i]\) does define an \(\NLC_{+}\) decomposition of  \(G[M,i]\). Let \(j = \ell_{i}(M)\), and let us consider the set
\[
\alpha_i^j(M) =  \{(\textsf{root}^j_i(u), \textsf{parent}^j_i(u),  \textsf{distance}^j_i(u)) \mid u \in \textsf{Prefix}(M,i-1) \; \text{and} \; \textsf{service}_i^j(u) \neq \bot\}.
\] 
Let \(u\) be a vertex such that \(\textsf{main}(u) \in \textsf{mainPrefix}(M,i)\). From condition~\ref{verif:5},  \(x_{i+1}(u)\) is of type \(\Join\). Hence, \(\textsf{Prefix}(\textsf{main}(u), i+1)\) induces a connected set of vertices in \(G\). Therefore, by conditions~\ref{verif:17}-\ref{verif:20} and~\ref{verif:26},  \(\alpha_i^j(M)\) certifies a Steiner tree in \(G\) rooted at \(\textsf{exit}_{i}(M)\), with set of terminals
\[
\textsf{terminals}_i(M) =  \{u \in V \mid  \exists  M' \in \textsf{mainPrefix}(M,i), \id(u)= \textsf{exit}_{i+1}(M')\}. 
\]  
From conditions~\ref{verif:21} and \ref{verif:24}-\ref{verif:27},   \(h_i(M)\)  is the homomorphism class of \({\Pi}\) obtained from the disjoint union function  \({\odot}_{\parallel}\) over the set  
\[
\{h_{i+1}(M'): M' \in \textsf{mainPrefix}(M,i) \}.
\] 
By the induction hypothesis, it follows that \(h_i(M)\) is the homomorphism class of \(G[M',i+1]\).  From conditions~\ref{verif:21}, \ref{verif:23},  and~\ref{verif:27}, we have that, for each \(s \in [k']\),
\[
\textsf{color}^s_{i}(M) = \sum_{u~ \in ~\textsf{terminals}_i(M)} \textsf{color}^{s}_{i+1}(u).
\]
Again, by the induction hypothesis, \(\textsf{color}_i^s(M)\)  is the number of vertices colored \(s\) in the root of~\(T[M,i]\).
\end{proof}

Finally, observe that by condition \(\ref{verif:6}\), for every $M$ and $M' \in \textsf{mainset}$, we have \(M_1 = M'_1\). Then, \(\textsf{Prefix}(M,1) = V\). Applying Lemma~\ref{lem:completeness} to an arbitrary \(M \in \textsf{mainset}\) and for \(i = 1\), we deduce that  \(T[M,1]\) is an \(\NLC_+\)-decomposition tree of \(G\). It follows that \(G\) is a graph of \(\NLC_+\)-width~\(k'\). Moreover, also by Lemma~\ref{lem:completeness}, we have \(h_1(M) = h(G)\).   Finally, by condition~\ref{verif:5.1} we conclude that \(G\) satisfies property \(\Pi\). \hfill\qed

\subsection{Certificate Size} 
For the  certificate size, let \(u\) an arbitrary vertex, let \(i \in d(u)\), and let \(j \in \{0,1\}\). Observe that each item of \(\textsf{main}_i(u)\), \(\textsf{aux}_i(u)\), and \(\textsf{service}^j_i(u)\) can be encoded with \(\mathcal{O}(\log n)\) bits.  By Lemma~\ref{lem:ldc}, we have that \(d(u)= \mathcal{O}(\log n)\). Therefore, the certificate of \(u\) is of size \(\mathcal{O}(\log^2 n)\).

%% file: optimization.tex

\section{Certifying Solutions of Optimization Problems}\label{se:optim}

In this section we extend our proof-labelling scheme to optimization problems expressible by  $\MSO_1$ predicates over graphs and vertex sets, on graphs of bounded clique-width. More precisely we consider $\MSO_1$ predicates $\Pi(G,X)$, over graphs $G=(V,E)$ and vertex subsets $X \subseteq V$. An optimization problem for $\Pi$ on graph $G$ consists in finding the maximum (or minimum) size vertex subset $X$ such that $\Pi(G,X)$ is true. For example, one can easily describe an  $\MSO_1$ predicate $\Pi(G,X)$ expressing that "$X$ is an independent set of $G$", or "$X$ is a dominating set of $G$". Hence problems like minimum dominating set or maximum independent set are particular cases of optimization problems expressible by  $\MSO_1$ predicates over graphs and vertex sets.

\begin{theorem}\label{theo:main-OPT}
Let $k$ be a non-negative integer. For every $\MSO_1$ predicate~$\Pi$ on vertex-labeled graphs and vertex-sets, and every $n$-node graph~$G=(V,E)$ with $\cw(G)\leq k$, and every $X\subseteq V$, there exists a PLS certifying that $X$ is the set of minimum (or maximum) size such that $(G,X)$ satisfies~$\Pi$ among all sets $Y$ such that $(G,Y)$ satisfies~$\Pi$, with $O(\log^2 n)$-bit certificates. 
\end{theorem}

Theorem~\ref{theo:main-OPT} actually also extends to $\MSO_1$ predicates $\Pi(G,X_1,X_2,\dots,X_q)$ over graphs and a bounded number of vertex subsets, and the optimization can be performed with respect to some linear evaluation function over the sizes of $X_1, \dots, X_q$. Such optimization problems are called $LinEMSOL$ in~\cite{CourcelleMR00}. However, since most natural optimization problems are expressed on a single vertex subset, and in order to keep the notations simpler, we focus here solely on predicates of type $\Pi(G,X)$.

\subsection{Proof of Theorem~\ref{theo:main-OPT}}
\label{app:theo:main-OPT}

Operations $\Join_S$ and $\mathsf{recolor}_R$ directly extend to pairs of graphs and vertex subsets: $(G_1,X_1) \Join_S (G_2,X_2)$ corresponds to pair $(G_1 \Join_S G_2, X_1 \cup X_2)$, and recoloring $\mathsf{recolor}_R(G,X)$ is simply the pair $(\mathsf{recolor}_R(G),X)$. 



The notion of regularity (Definition~\ref{de:reg}) extends to predicates over graphs and vertex subsets, moreover such $\MSO_1$ predicates are regular on graphs of bounded NLC-width. (Again, this even holds for predicates with a bounded number of vertex subsets~\cite{CourcelleMR00}).

\begin{definition}\label{de:regOpt}
Let $\Pi(G,X)$ be a predicate over graphs and vertex sets. Predicate $\Pi$ is called  \emph{NLC-regular} if, for any value $k$, we can associate a finite set $\cC$ of \emph{homomorphism classes} and a \emph{homomorphism function} $h$, assigning to each graph $G \in \NLC_k$ and each vertex subset $X$ of $G$ a class $h(G,X) \in \cC$ such that:
\begin{enumerate}
\item If $h(G_1,X_1) = h(G_2,X_2)$ then $\Pi(G_1,X_1) = \Pi(G_2,X_2)$.
\item For each composition operation $ \Join_S $ of arity 2 there exists a function $\odot_{\Join_S}: \cC \times \cC \rightarrow \cC$ such that, for any two $\NLC_k$ graphs $G_1$ and $G_2$ with vertex subsets $X_1$, $X_2$,
 $$h((G_1,X_1) \Join_S (G_2,X_2)) = \odot_{\Join_S }(h(G_1,X_1),h(G_2,X_2)).$$
 \item For each composition operation $\mathsf{recolor}_R$ of arity 1 there is a function $\odot_{\mathsf{recolor}_R}: \cC \rightarrow \cC$ such that, for any $\NLC_k$ graph $G$ and vertex subset $X$,
$$h(\mathsf{recolor}_R(G,X)) = \odot_{\mathsf{recolor}_R}(h(G,X)).$$
\end{enumerate}
Moreover functions  $\odot_{\Join_S}$ and $\odot_{\mathsf{recolor}_R}$ can be computed in constant time for any recoloring function $R:[k] \to [k]$ and any set $S$ of pairs of $[k] \times [k]$. Also, for any graph with a unique vertex colored $i \in [k]$, its two classes (corresponding to the two possible vertex subsets, the singleton and the empty set) can be computed in constant time. Again, by constant time, we mean only depending on $k$ and on the $\MSO_1$ formula describing $\Pi$.
\end{definition}

For example, if $\Pi(G,X)$ expresses that $X$ is an independent set of $G$, then we can take $\cC$ as the set of subsets of $[k]$, plus a special symbol $\bot$, and define  $h(G,X)$ as follows: $h(G,X) = \bot$ if $X$ is not an independent set, otherwise $h(G,X)$ is the subset of $[k]$ such that $i \in h(G,X)$ if $X$ contains some vertex colored $i$. 
Again, it is a matter of exercise to construct the classes for graph with a single vertex, as well as functions $\odot_{\mathsf{recolor}_R}$ and $\odot_{\Join_S}$. Indeed, if $G$ has a unique vertex $x$ coloured $i \in [k]$ then $h(G,\emptyset) = \emptyset$ and $h(G,\{x\}) = \{i\}$. Function $\odot_{\mathsf{recolor}_R}$ simply takes into account the recoloring of vertices. Let $R : [k] \to [k]$ be a recoloring. Note that $\odot_{\mathsf{recolor}_R}(\bot) = \bot$ (this corresponds to pairs $(G,X)$ where $X$ is not an independent set of $G$). Let $c$ be a subset of $[k]$, note that  $\odot_{\mathsf{recolor}_R}(c) = c'$ with $c' = \{j \in [k] \mid \exists i \in c \text{~such that~} R(i) = j\}$. Indeed, if $c=h(G,X)$ for some pair $(G,X)$, then $c'$ corresponds to the colours of vertices of $X$ after recoloring. Eventually, we describe the binary function $\odot_{\Join_S} : \cC \times \cC \to  \cC$. Observe that $\odot_{\Join_S}(c_1,c_2) = \bot$ if $c_1= \bot$ or $c_2 = \bot$. Indeed this corresponds to the case when $c_1 = h(G_1, X_1)$ and $c_2 = h(G_2, X_2)$, and $X_1$ is not an independent set of $G_1$, or $X_2$ is not an independent set of $G_2$ ; therefore $X = X_1 \cup X_2$ cannot be an intependent set of $G_1 \Join_S G_2$. For similar reasons, if there is some pair $(i,j)\in S$ such that $i\in c_1$ and $j \in c_2$, then $\odot_{\Join_S}(c_1,c_2) = \bot$, because the $\Join_S$ operation adds an edge between $X_1$ and $X_2$, so $G$ is not an independent set of $G$. Otherwise, we have $\odot_{\Join_S}(c_1,c_2) = c_1 \cup c_2$, since the set of colours of $X$ is the union of the sets of colours of $X_1$ and $X_2$. This proves the NLC-regularity of predicate ''$X$ in an independent set of $G$''.

By~\cite{CourcelleMR00}, all $\MSO_1$ predicates are regular:

\begin{proposition}\label{pr:regOpt}
Any  property $\Pi$ over graphs and vertex subsets expressible by an $\MSO_1$ predicate is NLC-regular. Moreover, given the corresponding $\MSO_1$ formula and parameter $k$, one can explicitely compute the set of homomorphism classes $\cC$, as well as  functions $\odot_{\Join_S}$ and $\odot_{\recolor_R}$ for any $S \in [k] \times [k]$ and any $R : [k] \to [k]$, and the homomorphism class of the graph formed by a unique vertex coloured $i$, for any $i \in [k]$.
\end{proposition}

To unify maximization and minimization problems, assume that we deal with weighted graphs. In this setting, each vertex $u$ of $G$ is assigned an integer $weight(u)$ in the interval $[-MAXW,+MAXW]$ for some positive integer constant $MAXW$. The aim of the optimization is to find a maximum weight vertex subset $X$ satisfying $\Pi(G,X)$. Natural minimization problems are turned into maximization problems, using negative weights.

In the certification context, we will be given graph $G$ together with the set $X$ of \emph{selected vertices}. The selected vertices are encoded by a function $sel : V \to \{0,1\}$. The set $X_s$ of selected vertices corresponds to  $X_s = \{u \in V(G) \mid sel(u) = 1\}$ and we aim to certify that $\Pi(G,X_s)$ is true and $X_s$ is of maximum weight among all sets $X$ satisfying $\Pi(G,X)$. Each node $u$ of the graph knows its status $sel(u)$, so it knows whether it is selected or not. Let us note that we prefer to separate the vertex selection function $sel$ and the vertex labeling function $\ell$, for two reasons: they play different roles (the labelling is part of the input, while selected vertices must be proved to be an optimal output for the problem), moreover the two can be combined, in the sense that the optimization is performed on labeled graphs (e.g., in the red-blue dominating set problem, the input is a graph with vertices labelled 0 for red and 1 for blue, and the aim is to find a minimum size set X of red vertices dominating all blue vertices).

We need to enrich the certificates  for decision problems described in the Section~\ref{se:cwd}. For each class $c \in \cC$, let $MaxW(G,c)$ denote the maximum $weight(X)$ over all vertex subsets $X$ of graph $G$ such that $h(G,X) = c$. If no such $X$ exists, we set $MaxW(G,c) = -\infty$. 
Informally, in each certificate, homomorphism classes over subgraphs of $G$ are replaced by homomorphism classes over the same subgraphs, with the corresponding subset of selected vertices. Moreover, whenever we had to consider a partial solution over a subgraph $G'$ of $G$, we also store in the message, besides the homomorphism class $h(G,X')$, the weight of the selected vertices $X'$ of $G'$, as well as all values $MaxW(G',c)$, over all possible classes $c \in \cC$. The latter will allow to check that there is no better solution for the problem than our set of selected vertices.

The following proposition is a straightforward application of the definition of $MaxW(G,c)$. It was also observed, in slightly different terms, in~\cite{CourcelleMR00} and it provides the tool to update values $MaxW$, from bottom to top, if we are given a decomposition tree. 
\begin{proposition}\label{pr:MaxW}
Let $\Pi(G,X)$ be an $\MSO_1$ predicate over graphs and vertex subsets,  $k$ be a fixed parameter, and $c \in \cC$ a homomorphism class.
\begin{enumerate}
\item If $G$ is formed by a single vertex $x$ of colour $i$. Then $$MaxW(G,c) = \max \{ weight(X)  \mid X \subseteq \{x\} \text{~such that~} h(G,X)=c\}.$$
\item If $G = G_1 \Join_S G_2$, then
$$MaxW(G,c) = \max_{c_1,c_2 \in \cC\times \cC}\{MaxW(G_1,c_1) + MaxW(G_2,c_2)\},$$
over all pairs of classes ${c_1,c_2 \in \cC\times \cC} \text{~s. t.~} \odot_{\Join_S}(c_1,c_2)=c$.
\item If $G = R(G')$, then
$$MaxW(G,c) = \max\{MaxW(G',c') \mid {c'\in \cC} \text{~such that~} \odot_R(c') = c.\}$$
\end{enumerate}
In all above equations we consider that the maximum of an empty set is $- \infty$.
\end{proposition}

We can now formally describe the certificates. We only precise what changes with respect to messages defined in Section~\ref{se:cwd}.

\subparagraph{Main messages.} 

In the decision setting, the main message of each node $u$ contains a sequence $\textsf{h}(u) = (h_1(u), \dots, h_d(u))$. Recall that $x_i(u)$ denotes the $i$th node on the path of the $NLC_+$ decomposition tree $T$ of $G$, from the root to the leaf creating $u$, and $h_i(u)$ denotes the homomorphism class of predicate $\Pi$ at node $x_i(u)$, i.e., corresponding to the induced subgraph $G[x_i(u)]$ of $G$. In the optimization setting we simply use as $h_i(u)$ the class $h(G[x_i(u)],X_s[x_i(u)])$, where $X_s[x_i(u)]$ is the set of selected vertices of $G[x_i(u)]$.

We add to the main message of $u$ the sequence $\textsf{weight}(u) = (\textsf{weight}_1(u), \dots, \textsf{weight}_d(u))$, such that, for each \(i \in \{1, \dots, d\}\), \(\textsf{weight}_i(u)\) denotes the total weight of selected vertices at $x_i(u)$, that is, the weight of $X_s[x_i(u)]$.

We also add the sequence $\textsf{MaxW}(u) = (\textsf{MaxW}_1(u), \dots, \textsf{MaxW}_d(u))$, where  $\textsf{MaxW}_i(u)$ denotes the list  $\{MaxW(G[x_i(u)],c) \mid c \in \cC\}$ of maximal weights of possible partial solutions at node $x_i$, over all possible classes $c \in \cC$. 

\subparagraph{Auxiliary messages.} 

Recall that at each node $x$ of type $\Join$ of the NLC+ decomposition tree, we use a spanning tree $A(x)$ of the graph $G[x]$ hanging at $x$ and and auxiliary messages to check that $G[x]$ is indeed connected, and that the messages in this subgraph are coherent. Formally each node $u$ of $G$ is provided an auxiliary message $\textsf{aux}(u) = (\textsf{aux}_1(u), \dots, \textsf{aux}_{d}(u))$, where $\textsf{aux}_i(u)$ is set to $\bot$ if $x_i(u)$ is of type $\parallel$, otherwise it contains the information required for the spanning tree $A(x_i(u))$ and the coherence of messages of $G[x_i(u)]$. The  difference for optimization problems consists in enriching the part \(\textsf{childrenMain}^j_i(u)=(z^j, h(z^j), \textsf{color}(z^j), \textsf{exit}(z^j))\) for \(j \in \{0,1\}\), where $z^0$ and $z^1$ denote the children of $x_i(u)$ in the NLC+ decomposition tree. As for main messages, in the optimization setting $h(z^j)$ will denote the homomorphism class for property $\Pi$ on $G[z^j]$ and its selected vertices. Moreover we add to $\textsf{childrenMain}^j_i(u)$ the information $\textsf{weight}(z^j)$ and $\textsf{MaxW}(z^j)$, for $j \in \{0,1\}$, where $\textsf{weight}(z^j)$ is the weight of selected vertices of $G[z^j]$ and $\textsf{MaxW}(z^j) = \{MaxW(G[z^j],c) \mid c \in \cC\}$ is the list of all optimal solutions, over all homomorphism classes, in graph $G[z^j]$. 

\subparagraph{Service messages.} 

Recall that these messages are created at each  node $x$ of type $\parallel$ in the decomposition tree, and they concern the vertices of some tree $S(x)$, connecting the graphs corresponding to subtrees of children $z_i$ of $x$, $1 \leq i \leq p$ through their exit vertices. Also recall that for each node $u$ of $S(x)$, we have defined the set of indices $\textsf{inCharge}_x(u) = \{i \in \{1, \dots, p\} \mid \textsf{exit}(z_i) \in S(x,u) \}$, meaning that the service message of $u$ for $x$ concerns the partial solution  \(G[ \textsf{inCharge}_x(u) ] \), the subgraph of \(G\) induced by the disjoint union of all graphs \(G[z_i]\) with \(i \in \textsf{inCharge}_x(u)\). By symmetry, let $X_s[ \textsf{inCharge}_x(u) ]$ be the union of labelled vertices $X_s[z_i]$,  \(i \in \textsf{inCharge}_x(u)\).

Recall that, for decision problems (see Section~\ref{se:cwd} for a full description of service messages), we used value $\textsf{class}^j_i(u)$ corresponding to homomorphism of $G[\textsf{inCharge}_{x}(u)]$ (here $j=0$ if $x$ is in the left child of $x_{i-1}(u)$, and  $j=1$ if $x$ is the right child of $x_{i-1}(u)$). In the optimization version, we use instead
 the homomorphism class \(h(G[ \textsf{inCharge}_x(u) ], X_s[ \textsf{inCharge}_x(u) ])\), and we add the weight $\textsf{weightService}^j_i(u)$ of $X_s[ \textsf{inCharge}_x(u) ]$, and  the list $\textsf{MaxWService}^j_i$ of all values $\{MaxW(G[ \textsf{inCharge}_x(u) ],c) \mid c \in \cC\}$.

We emphasize again that, when we generalized the certification messages from decision problems to optimization problems, we applied the same process on every partial solution, on (induced) subgraphs $G'$ of $G$: the homomorphism class $h(G')$ was replaced with the corresponding homomorphism class $h(G',X'_s)$ over selected vertices of $G'$, and the messages were enriched with the weight $weight(X'_s)$ of $X'_s$ and with the set $\{MaxW(G',c) \mid c \in \cC\}$ of all maximal partial solutions of $G'$, over all possible homomorphism classes.

Therefore, in the verification scheme, we only need to update all verifications that were in charge of homomorphism classes $h(G')$ for some subgraph $G'$, and to enrich them to verify:
\begin{itemize}
\item $h(G',X'_s)$, using Definition~\ref{de:regOpt} instead of Definition~\ref{de:reg},  
\item  the weight $weight(X'_s)$ of $X'_s$, by simply summing over smaller partial solutions,
\item the wholes set $\{MaxW(G',c)\}$ over all classes $c$ of $\cC$. For the latter, we simply use Proposition~\ref{pr:MaxW}, in the same manner as we used Definition~\ref{de:reg}, to check that each value $MaxW(G',c)$ is correct.
\end{itemize}
This occurs in the following places -- we refer again to modifications w.r.t. Section~\ref{se:cwd}.

\subparagraph{Verification of main messages.} 

Item~\ref{verif:4}, on graph $G'$ with a single vertex $u$: we must check that $weight(X'_s) = 0$ if $u$ is not selected, otherwise $weight(X'_s) = weight(u)$. Values $MaxW(G',c)$ are computing using the first item of Proposition~\ref{pr:MaxW}.

Item~\ref{verif:5.1}, for the whole graph $G$. Besides checking that $h_1(u) = h(G,X_s)$ is an accepting class for property $\Pi$, we also check that the weight of $X_s$ (i.e., $\textsf{weight}_1(u)$) equals the maximum of $MaxW(G,c) \in \textsf{MaxW}_1(u)$ over all \emph{accepting} classes $c$. This ensures that the solution  $X_s$ is indeed optimal.

Item~\ref{verif:9}, corresponding to an edge between vertices $u$ and $v$ such that  $x_{i^*}(u)$, their lowest common ancestor in the NLC+ decomposition tree, is a binary $\Join$ node. In this case $G[x_{i^*}(u)]$ is obtained from $G[x_{i^*+1}(u)]$ and $G[x_{i^*+1}(v)]$ by a sequence of operations $\Join_S$ then $\recolor_R$. As for decision problems, we verify using Definition~\ref{de:regOpt} that class $h_{i^*}(u)$ is coherent with classes $h_{i^*+1}(u)$ and $h_{i^*+1}(v)$ after performing  $\Join_S$ then $\recolor_R$. Then we perform the similar verification for each value $MaxW(G[x_{i^*}(u)],c) \in \textsf{MaxW}_{i^*}(u)$ using Proposition~\ref{pr:MaxW} (items 2 and 3) to check that $\textsf{MaxW}_{i^*}(u)$ is indeed obtained from $\textsf{MaxW}_{i^*+1}(u)$ and $\textsf{MaxW}_{i^*+1}(v)$ by a sequence of operations $\Join_S$ then $\recolor_R$. Also, we check that $\textsf{weight}_{i^*}(u) = \textsf{weight}_{i^*+1}(u) + \textsf{weight}_{i^*+1}(v)$.

\subparagraph{Verification of auxiliary messages.} We do not need any changes, since this is only a coherence check based on the equality of parts $\textsf{childrenMain}$ of the messages.

\subparagraph{Verification of service messages.}  We need to verify messages $\textsf{class}^j_i(u)$ corresponding to homomorphism of $G[\textsf{inCharge}_{x}(u)]$, the weight $\textsf{weightService}^j_i(u)$ of $X_s[ \textsf{inCharge}_x(u) ]$, and  the list $\textsf{MaxWService}^j_i(u)$ of all values $\{MaxW(G[ \textsf{inCharge}_x(u) ],c) \mid c \in \cC\}$. 

Item~\ref{verif:21} is performed when $u$ is a leaf of the  tree \(S(x_i^j)\), in particular it is the exit vertex of $G[x_{i+1}(u)]$ and $ \textsf{inCharge}_x(u)  = G[x_{i+1}(u)]$. As in Section~\ref{se:cwd} we check that \(\textsf{class}_i^j(u) = h_{i+1}(u)\). We also check that $\textsf{weightService}^j_i(u) = \textsf{weight}_{i+1}(u)$ and $\textsf{MaxWService}^j_i = \textsf{MaxW}_{i+1}(u)$.

Recall that, for vertex $u$, \(\textsf{children}_i^j(u)\) denotes the set of nodes  \(w\in V\) such that \({\textsf{parent}_i^j(w)=u}\), and this set is known to $u$. In particular $ \textsf{inCharge}_x(u)  = \parallel_{w \in \textsf{children}_i^j(u)} \textsf{inCharge}_x(w)$. In particular $ \textsf{inCharge}_x(u)  = \cup_{w \in \textsf{children}_i^j(u)} \textsf{inCharge}_x(w)$.

Item~\ref{verif:24} This verification occurs when  \(\textsf{children}_i^j(u)\) has a unique node $w$. Besides the verification that \(\textsf{class}_i^j(u) = \textsf{class}_i^j(w)\) we also check that $\textsf{weightService}^j_i(u) = \textsf{weightService}^j_i(w)$ and $\textsf{MaxWService}^j_i(u) = \textsf{MaxWService}^j_i(w)$.

Item~\ref{verif:25} This happens when  \(\textsf{children}_i^j(u)\) has two or more nodes $w$. Recall that, in this case, $G[\textsf{inCharge}_{x}(u)] = \parallel_{w \in \textsf{children}_i^j(u)} G[\textsf{inCharge}_x(w)]$. As for decision problems, we order arbitrarily the vertices $w_1 \dots, w_p$ of \(\textsf{children}_i^j(u)\) and for each graph $H_q = G[\textsf{inCharge}_x(w_1)] \parallel \dots \parallel G[\textsf{inCharge}_x(w_q)]$, $1 \leq q \leq p$, we compute  $h(H_q,X_q)$, $weight(X_q)$ and sets $\{MaxW(H_p,c) \mid c \in \cC\}$ incrementally, by increasing values of $q$. Here $X_q$ denotes the set of selected vertices of $H_q$. Observe that $H_1 = G[ \textsf{inCharge}_{x}(w_1)]$ so node $u$ has all the required information in the message of $w_1$. For each $q, 2 \leq q \leq p$, these parameters are computed thanks to relation $H_q = H_{q-1} \parallel G[\textsf{inCharge}_x(w_q)]$, based on Definition~\ref{de:regOpt} and Proposition~\ref{pr:MaxW} adapted to operation $\parallel$, from the similar parameters on $H_{q-1}$ (previously computed at node $u$) and $G[\textsf{inCharge}_x(w_q)]$ (retrieved from the message of $w_q$. Eventually, $u$ checks that these parameters on $H_p$ correspond to $\textsf{class}^j_i(u)$, $\textsf{weightService}^j_i(u)$ and $\textsf{MaxWService}^j_i(u)$.

Item~\ref{verif:27} applies when  \(u = \textsf{exit}_{i}(u)\). In this case, $u$ checks for each \(j \in \{0,1\}\) that  \(\textsf{class}_i^j(u) = h_{i}(u)\), $\textsf{weightService}^j_i(u) = \textsf{weight}_{i}(u)$ and $\textsf{MaxWService}^j_i = \textsf{MaxW}_{i}(u)$.

The soundness and correctness proofs are, modulo the use of Definition~\ref{de:regOpt} and Proposition~\ref{pr:MaxW}, very similar the case of decision problems. Note that the new messages are still of size $O(\log^2 n)$.  \hfill\qed

%% file: conclusion.tex

\section{Conclusion}
\label{sec:conclusion}

In this paper, we have shown that, for every $\MSO_1$ property $\Pi$ on labeled graphs, there exists a PLS for $\Pi$ with $O(\log^2 n)$-bit certificates for all $n$-node graphs of bounded clique-width. This extends previous results for smaller classes of graphs, namely graphs of bounded tree-depth~\cite{FeuilloleyBP22}, and graphs of bounded tree-width~\cite{FraigniaudMRT22}. Our result also  enables to establish a separation, in term of certificate size, between certifying $C_4$-free graphs and certifying $P_4$-free graphs.  

A natural question is whether the certificate size resulting from  our generic PLS construction is optimal. Note that one log-factor is related to the storage of IDs, and of similar types of information related to other nodes in the graph. It seems hard to avoid such a log-factor. The other log-factor is however directly related to the depth of the NLC-decomposition, and our PLS actually uses certificates of size $O(d\cdot\log n)$ bits for graphs supporting an NLC-decomposition of depth~$d$. Nevertheless, the best generic upper bound for the depth~$d$ of an NLC-decomposition preserving bounded width is $O(\log n)$. This log-factor seems therefore hard to avoid too. Establishing the existence of a PLS for $\MSO_1$ properties in graphs of bounded clique-width using $o(\log^2n)$-bit certificates, or proving an $\Omega(\log^2n)$ lower bound on the certificate size for such PLSs  appears to be challenging. 

Another interesting research direction is whether our result can be extended to $\MSO_2$ properties. It is known that the meta-theorem from~\cite{CourcelleMR00} does not extend to $\MSO_2$. Nevertheless, this does not necessarily prevent the existence of compact PLSs for $\MSO_2$ properties on graphs of bounded clique-width. For instance,  Hamiltonicity is an $\MSO_2$ property that can be easily certified in \textit{all} graphs, using certificates on just $O(\log n)$ bits. Is there an $\MSO_2$ property requiring certificates of $\Omega(n^\epsilon)$ bits, for some $\epsilon>0$, on graphs of bounded clique-width? Finally, it might be interesting to study the existence of compact distributed interactive proofs~\cite{KolOS18} for certifying $\MSO_1$ or even $\MSO_2$ properties on graphs of bounded clique-width. Note that the generic compiler from~\cite{NaorPY20} efficiently applies to sparse graphs only whereas the family of graphs with bounded clique-width includes dense graphs.